\documentclass[a4, 11pt, fleqn ]{article}

\usepackage{amsmath,amssymb,amsfonts,amsthm,bbm,mathrsfs,verbatim} 
\usepackage[misc]{ifsym}
\usepackage{graphics}                 
\usepackage{graphicx}
\usepackage{caption}
\usepackage{subcaption}
\usepackage{float}
\usepackage{array}
\usepackage{color}                    
\usepackage{hyperref}                 
\usepackage{makeidx}
\usepackage[round]{natbib}
\usepackage{geometry}  
\usepackage{setspace}
\usepackage{bm}
\usepackage{authblk}

\usepackage{fancyhdr}

\newtheorem{theorem}{Theorem}[section]
\newtheorem{proposition}[theorem]{Proposition}

\newtheorem{lemma}[theorem]{Lemma}
\theoremstyle{definition}
\newtheorem{remark}[theorem]{Remark}
\newtheorem{definition}[theorem]{Definition}
\newtheorem{example}[theorem]{Example}

\def\R{{\mathbb R}}

\def\H{{\mathcal H}}
\def\M{{\cal M}}
\def\P{{\mathcal P }}

\def\J{{\cal J}}
\def\F{{\mathcal F }}

\newcommand{\ud}{\, {\rm d}}
\newcommand{\nnb}{\nonumber}

\def\RF{\rho_{F}}     
\def\RW{\rho_{W}}  
\def\RE{\rho_{E}}  
\def\d{{\mathsf d}}

\def\is{q}   

\def\tm{ {\widetilde m} }
\def\tS{{\widetilde S} }

\def\Y{{\cal Y}}
\def\Z{{\cal Z}}
\def\H{{\cal H}}

\newcommand\blue[1]{\textcolor{blue}{#1}}

\def\Zob{{Z^{(1)}}}
\def\Ztb{{Z^{(2)}}}
\def\Zo{ Z^{(1)}}
\def\Zt{ Z^{(2)}}

\title{Convex Risk Measures for the Aggregation of Multiple Information Sources and Applications in Insurance }
\date{}

\author[a]{\small G. I. Papayiannis}
\author[b]{\small A. N. Yannacopoulos}

\affil[a]{\footnotesize Hellenic Naval Academy, Department of Naval Sciences, Section of Mathematics and Mathematical Modelling and Applications Laboratory, Piraeus, Greece; Athens University of Economics \& Business, Stochastic Modelling and Applications Laboratory, Athens, Greece}
\affil[b]{\footnotesize Athens University of Economics \& Business, Department of Statistics and Stochastic Modelling and Applications Laboratory, Athens, Greece}

\begin{document}
	
\fancypagestyle{plain}{%
	\fancyhead[R]{\blue{Published in Scandinavian Actuarial Journal \url{https://doi.org/10.1080/03461238.2018.1461129}}} 
	\renewcommand{\headrulewidth}{0pt}}
	
\maketitle


\begin{abstract}
We propose a novel class of convex risk measures, based on the concept of the Fr\'echet mean, designed in order to handle uncertainty which arises from multiple information sources regarding the risk factors of interest. The proposed risk measures robustly characterize the exposure of the firm, by filtering out appropriately the partial information available in individual sources into an aggregate model for the risk factors of interest. Importantly, the proposed risks can be expressed in closed analytic forms allowing for interesting qualitative interpretations as well as comparative statics and thus facilitate their use in the everyday risk management process of the insurance firms. The potential use of the proposed risk measures in insurance is illustrated by two concrete applications, capital risk allocation and premia calculation under uncertainty.
\end{abstract}

\noindent {\bf Keywords:} Model Uncertainty; Model Aggregation; Convex Risk Measures; Robustness; Fr\'echet Variance;

\section{Introduction}

Financial and insurance risk managers often face situations where they have to make decisions concerning risks whose exact probability distribution is unknown on account of either incomplete data or information, or because they are connected with subjective beliefs towards risk factors  or multiple information sources of perhaps unknown reliability. Such situations can be abstracted by considering the risk as a random variable $X:\Omega\to\R$, with the true probability measure, $\mu_0$, is unknown to the risk manager who is only aware of a family $\mathcal{M}$ consisting of possible alternative probability models $\mu_i$, $i \in {\cal I}$, where ${\cal I}$ is an index set, for the possible scenarios on $\Omega$. These alternative models may come from various sources (e.g. by different experts or may arise as partial models consistent with incomplete data available to the individual sources) and are assumed to be approximations to the true probability model $\mu_0$ for the different states of the world in $\Omega$. The success of each approximation $\mu_i$ for $\mu_0$, depends on the reliability of the source or the expert's ability to model the risk under question. The important issue arising is which of the potential probability models in $\mathcal{M}$ should be used by the risk manager in order to assess the risk $X$ and what is an appropriate risk measure under which the policy decisions (e.g. risk premia calculation or capital allocation policies) will be made. 

Situations as described above are abundant in insurance; as a simple concrete example consider the problem of evaluating the total risk position $X$ of a multi-sector insurance firm consisting of $d$ sectors, each contributing the risk $X_{\is}$, $\is=1,\ldots, K$ to the total risk of the firm $X=\sum_{\is=1}^{K} X_{\is}$. In this case, each information source can be thought of as the manager of each sector of the firm who is well aware of the scenarios $\Omega_{\is} \subset \Omega$ that affect the risk of her sector $X_{\is}$, but may be unaware of other scenarios $\Omega_j \in \Omega$  that affect the risks $X_{\is}$ for $\is \ne i$, which are needed to evaluate the total risk $X$. As the modelling of dependence can be quite difficult there is uncertainty in the dependence structure between the various risks, thus leading to a multitude of probability models $\M$ which can be used to estimate the risk induced by the random variable $X$. 

A very convenient way of treating situations as  described above is through the use of convex risk measures. The class of convex risk measures (see \cite{follmer2002robust}) as an extension of the class of coherent risk measures (first introduced in \cite{artzner1999coherent}) has become widespread in theoretical and applied insurance research and it is expected to gradually become the industry standard in insurance within the Solvency framework, as more and more often  on account of its various drawbacks the use of VaR as a mean of quantifying risk  is replaced by the use of its coherent counterparts such as CVaR or AVaR. In order to cope with risk management in situations where multiple competing models are formulated for the various scenarios in $\Omega$, a convex risk measure should be utilized which is capable of combining efficiently the information provided by all partial models into an aggregate model,  filtering out possible inconsistencies of the models in $\mathcal{M}$, in order to formulate the most appropriate and robust characterization of the risk $X$ and formulate policies accordingly.

This paper proposes a new class of convex risk measures, the Fr\'echet risk measures, that can be used for this purpose, and which optimally combine all (partial) information sources into a single model and provide reliable and robust with respect to model uncertainty estimates to the risk position of the firm. The construction of the proposed risk measures requires a reliable measure for the distance between different probabilistic models (corresponding to different opinions, estimates of information sources) and since all models are assumed to be representations of the same phenomenon, high dispersion of models must be penalized.  However, such a task requires the generalization of the concept of dispersion for the space of probability measures (representing here possible probability models), which is not a vector space thus invalidating the usual definition of the mean and consequently the dispersion.  To this end we treat the space of probability measures (the natural topological space setting for probability models) as a metric space with a chosen metrization which allows us to extend the concept of  the mean and the dispersion in terms of the concept of the Fr\'echet mean and the Fr\'echet  dispersion which are then defined through a variational  formulation. We can then use the robust representation of convex risk measures, due to \cite{follmer2002robust} and \cite{frittelli2002putting}, combined with a properly chosen penalty function on the set of probability measures. This penalty function is related with the Fr\'echet variance of the set of priors $\mathcal{M}$, and penalizes extreme scenarios. In this way, we obtain the a novel class of risk measures which can be used in order to filter out uncertainty of the probabilistic models of the risk under consideration and provide a robust and easily computable estimation of the total risk incurred.

Clearly, the exact form of the proposed Fr\'echet risk measures depends on the choice of metrization for the space of probability measures. One particularly popular metrization which is also adopted in this paper is the Wasserstein metric, while the use of pseudo-metrics such as the Kullback-Leibler divergence is also a popular choice which is also investigated in this work. We provide versions and explicit calculations of the Fr\'echet risk measures for such choices, with a special emphasis on the case of the Wasserstein metric for which recent research has shown that it is  a very successful metric for the space of probability measures. For example, the Wasserstein metric has been used as a loss functional for learning schemes and extensive empirical and theoretical studies have shown that the use this metric as a distance functional provides superior results as compared to other metrics or pseudo-metrics such as Kullback-Leibler divergence (see e.g. \cite{py2016a}), in the space of probability measures. Furthermore, the extension of the mean in the metric space of probability measures, called the Wasserstein barycenter, has been used as a tool for model selection with very satisfying results (\cite{py2016a}). As already mentioned, the proposed risk measures in many cases can be estimated explicitly in semi-analytic form, thus providing a very convenient tool for practitioners to evaluate and analyse risk positions. Their use in insurance is illustrated using two indicative applications, in the problem of capital allocation and in insurance premia calculations.



\section{  Fr\'echet Risk Measures : Definitions and Properties}\label{Section2}

In this section we introduce the concept of Fr\'echet risk measures, a novel concept that may find important applications in a variety of concrete risk management problems, related to multi-agent decisions or multiple and possibly conflicting priors concerning the probability measure governing the phenomenon.

\subsection{\textit{\textbf{ The General Class of Fr\'echet Risk Measures}}}\label{CFRM}

Let $Z = (Z_1,\ldots, Z_d)$ be  a  set of common  stochastic risk factors affecting all possible ventures or positions of the firm. We identify the set of possible states of the world (i.e. states of the economy) with $\R^d$ and hereafter consider $\Omega = \R^d$. We assume that the description of the economy through the stochastic factors $Z$ is sufficient in the sense that for any possible position of the firm $X$ there exists a Borel measurable map $\Phi_0:\Omega = \R^d\to \R$ such that $X=\Phi_0(Z)$ \footnote{Or as is sometimes more common in notation $-X=\Phi_0(Z)$.}. Therefore knowledge of the state of the stochastic factor $Z$ provides a complete description of the actual value of the position $X$  and using the composite mapping $\Phi_0\circ Z$ the position of the firm can be considered as a random variable $X:\Omega=\R^d \to \R$. This approach is the standard approach to the modern theory of quantitative risk management (see e.g. \cite{mcneil2015quantitative}).

We will assume that a number of possible priors are available concerning the distribution of the stochastic risk factors $Z$, in terms of probability measures on ${\cal B}(\R^d)$, the Borel $\sigma$-algebra on $\Omega=\R^{d}$. We will denote the set of these priors by $\M=\{\mu_i\}_{i =1, \ldots, n}$, which for the sake of simplicity, are assumed to be absolutely continuous with respect to the multidimensional Lebesgue measure. Each model may contain parts of the reality and perhaps better represents particular aspects of the distribution of $Z$, while none of these may in principle be the true model for $Z$. In some sense $\M$ contains the available information for $Z$ which may be partial due to incomplete observations, etc. Therefore, when trying to assign a model for the description of $Z$ we cannot simply single out a particular model out of $\M$ but rather we need to use an appropriate combination of the elements of $\M$, designed in such a way that the partial information contained in each model in $\M$ is aggregated so as to get as complete as possible a probabilistic description for $Z$. Naturally, the uncertainty on the probability measures concerning $Z$ reflects on uncertainty regarding the probability measures concerning the position of the firm $X$.

To deal with model uncertainty concerning $Z$ one can use the concept of convex risk measures and in particular employ results on their robust representation, in order to assign to each risk position $X$ a real number quantifying the risk, constructing a functional $\rho:\mathcal{L} \to \R$ where $\mathcal{L}$ is a suitable vector space containing all possible random variables $X:\Omega \to \R$ under consideration (where $\Omega = \R^d$). This mapping must satisfy the axiomatic framework of convex risk measures (see \cite{follmer2002robust}), but at the same time must also take into account the fact that there are multiple plausible models, none of which can in principle a priori be discarded as a description of the stochastic factors $Z$. In this section we will propose a class of convex risk measures, the Fr\'echet risk measures, which may provide a useful framework for the quantification of risk in situations as the ones described above.

Dispersion in model space (which for the purpose of this work coincides with the space of probability measures) can be quantified using the concept of {\it Fr\'echet variance}, in the space of probability measures ${\cal P}(\Omega)$ appropriately metrized by a metric $\d$. The concept of the Fr\'echet variance (and the related Fr\'echet mean) has been first introduced by Maurice Fr\'echet (see e.g. \cite{frechet1948elements} and \cite{izem2007analysis} who also provided a useful decomposition of the Fr\'echet variance), in an attempt to generalize the concepts of the variance and the mean for random variables taking values not on a vector space but rather in general metric spaces not admitting a linear structure (as for example the space of probability measures). This concept has its roots in differential geometry but over the last years has found many applications in various areas such as decision theory, probability theory, functional data analysis, image processing and other fields. This notion of mean is based upon a variational argument, and in particular on the minimization of the {\it Fr\'echet function} which for any set of weights $w=(w_1,\ldots, w_n) \in \Delta^{n-1}$, where $\Delta^{n-1}$ is the positive unit simplex of $\R^n$, is defined as the mapping ${\mathbb F}_{\M} : {\cal B}(\R^d) \to \R_{+}$,
\begin{eqnarray*}
{\mathbb F}_{\M}(\mu):=\sum_{i=1}^{n}w_{i}\d^2(\mu,\mu_i),
\end{eqnarray*}
where $\d$ is a metric on the space ${\cal P}(\R^d)$, the space of probability measures on $\Omega=\R^d$. The minimum of this function over ${\cal P}(\R^d)$ is called the {\it Fr\'echet variance} of the set of priors $\M$ and will be denoted by
\begin{eqnarray*}
V_{\M} := \min_{\mu \in {\cal P}(\R^d)} {\mathbb F}_{\M}(\mu).
\end{eqnarray*}
We will hereafter use the notation $F_{\M}$ for the {\it normalized Fr\'echet function}
\begin{eqnarray}\label{mff}
F_{\M}(\mu)={\mathbb F}_{\M}(\mu) - V_{\M} \ge 0.
\end{eqnarray}
The normalized Fr\'echet function is a positive valued  convex function on the space of probability measures. The minimizer 
\begin{eqnarray}\label{fr-mean}
\mu_{B} := \arg\min_{\mu \in\mathcal{P}(\R^d)}F_{\M}(\mu)=  \arg\min_{\mu \in\mathcal{P}(\R^d)}  \sum_{i=1}^n w_i \d^2(\mu, \mu_i),
\end{eqnarray}
is called the {\it Fr\'echet mean} (or the barycenter) of $\M$, whereas the {\it Fr\'echet variance} of $\M$  is defined as $V_F(\M):=F_{\M}(\mu_{B})$. The Fr\'echet mean may not be unique, however, for the class of measures considered here uniqueness can be proved (see e.g. \cite{afsari2011riemannian,arnaudon2013medians, kroshnin2017fr}).

The Fr\'echet mean  $\mu_B$ can be considered as the aggregate model for the particular set of priors $\M$ and is the appropriate generalization of the least-square estimator (in the sense of the metric $\d$) in the space of probability measures. The choice of the weight vector depends on the credibility of the source from which each prior originates from; we may choose all weights equal if all priors are equally acceptable or if no knowledge about the reliability or the performance of the competing models exists. Other choices are possible, reflecting variability of the credibility of the various priors, based perharps on experience as to their performance (e.g. using a learning procedure similar to that  presented in \cite{py2016a}). The Fr\'echet mean also depends on the choice of metric. In view of the discussion in the introduction, an appropriate choice for the metric $\d$ is the Wasserstein metric and most of the results in this work will focus on the risk measures induced by this choice of metric, and how the aggregation of models affects the determination of the risk assigned to the position in terms of the chosen risk measure. 

We recall the  general robust representation of convex risk measures
\begin{equation}\label{robust}
\rho(X) = \sup_{\mu\in\mathcal{P}(\Omega)} \left(\mathbb{E}_{\mu}[-X] - a(\mu) \right),
\end{equation}
where $\mathcal{P}(\Omega)$ is the set of probability measures  on the set of states of the world $\Omega$ \footnote{ identified here with $\R^{d}$}, and $a :{\cal P}(\Omega) \to \R$ a convex  penalty function that penalizes certain extreme scenarios,
proposed in \cite{follmer2002robust} and \cite{frittelli2002putting}. Since our aim is in filtering out the information concerning the risk $X$ available in the set of priors for the stochastic factors $Z$, $\M$, we propose to use the Fr\'echet function $F_{\M}$ as a part of the penalty function in the robust representation \eqref{robust}, so that the penalty term will effectively select as a minimizer in problem \eqref{robust} a probability measure which should be close to the Fr\'echet mean $\mu_B$ of the set of priors $\M$. We therefore introduce the novel class of Fr\'echet risk measures.

\begin{definition}[{\bf Fr\'echet risk measure}]\label{frisk} 
Let $\alpha : \R \to \R_{+}\cup\{\infty\}$ be an increasing function, such that $\alpha(0)=0$, and let $\Phi_0  : \R^{d} \to \R_{+}$ be the risk mapping connecting the stochastic factors $Z$ to the risk position $X$ of the firm. We define the {\it Fr\'echet risk measure} for any $\gamma\in(0,\infty)$ as
\begin{eqnarray}\label{frechet-risk}
\RF(X) := \sup_{\mu \in \mathcal{P}(\R^d)} \left\{ \mathbb{E}_{\mu}[-X] -\frac{1}{2 \gamma} \alpha ( F_{\M}(\mu) ) \right\},
\end{eqnarray}
where  $\M$ is the set of priors for  $Z$, $-X=\Phi_0(Z)$ and $F_{\M}$ the normalized Fr\'echet function defined in \eqref{mff}.
\end{definition}

The following proposition collects some properties of  Fr\'echet risk measures. We will use the explicit notation $\RF(X)=\RF(X ;\gamma)$ to emphasize the dependence on the parameter $\gamma$.

\begin{proposition}[Properties of  Fr\'echet risk measures]\label{Prop2.2} \hfill
Consider the measurable space $(\Omega,\mathcal{F})$ where here $\Omega = \R^d$ and $\mathcal{F} \subseteq \mathcal{B}(\R^d)$ a given $\sigma$-algebra.
\begin{itemize}
\item[(i)] The Fr\'echet risk measures considered as mappings $\RF : {\mathcal L} \to \R_{+}$, where ${\mathcal L} := \{ X:\Omega \to \R: X\,\, \mathcal{F}-\mbox{measurable }\}$ are convex risk measures, such that $\RF(X ;\gamma) \ge {\mathbb E}_{\mu_{B}}[-X]$, for all $\gamma \ge 0$, where $\mu_{B}$ is the Fr\'echet barycenter of $\M$ (defined in \eqref{fr-mean}).
\item[(ii)] For any fixed $-X=\Phi_0(Z)$, it holds that $\RF(X ; \gamma_1) \le \RF(X ;\gamma_2)$ for any $0 \le \gamma_1 \le \gamma_2$, while
\begin{eqnarray*}
\lim_{\gamma \to 0^{+}} \RF(X ;\gamma) = {\mathbb E}_{\mu_{B}} [ \Phi_0(Z)]  \le \RF(X ; \gamma) \le {\rm ess}\sup_{Z \in \Omega=\R^{d}}  \Phi_0(Z),
\end{eqnarray*}
where $\mu_{B}$ is the Fr\'echet barycenter of $\M$.
\end{itemize}
\end{proposition}

\begin{proof}
For the proof of the Proposition please see Appendix \ref{A.1}.
\end{proof}

The above proposition implies that the parameter $\gamma$ plays the role of an uncertainty aversion parameter, and has a similar interpretation as the relevant parameter in the entropic risk measure (see e.g. \cite{follmer2011entropic, ahmadi2012entropic}). Furthermore, if $-X$ is considered as a loss so that $-X\geq 0$ a.s. we have that $\rho_F(X;\gamma)\geq 0$ for every $\gamma \in (0,\infty)$.

In view of the discussion above, this class of convex risk measures filters out effectively multiple information and uncertainty and provides reliable representation of the risk using an aggregate model. This type of risk measures can also be characterized as \emph{minimum variance} risk measures, in the sense that they are robust risk measures centered around the corresponding Fr\'echet mean of the set of priors. Because of the geometric interpretation of the Fr\'echet mean, in some sense, the family of Fr\'echet risk measures introduced here can be called a {\it barycentric} risk measure.


The freedom in the choice of $\alpha$ allows for a wide variety of risk measures through Definition \ref{frisk}. For example, if $\alpha$ is chosen so that $\alpha(F_{\M}(\mu))={\mathbb I}_{C}(\mu)$, the indicator function (in the sense used in convex analysis i.e. ${\mathbb I}_{C}(x) = 0$, if $x\in C$ and ${\mathbb I}_{C}(x) = \infty $ if $x\notin C$ ) of a convex set $C \subset \P(\Omega)$, e.g. $C=\{ \mu \in \P(\R)  \, : \, F_{\M}(\mu) \le \zeta\}$, we obtain the subclass of constraint risk measures. A particular case of interest is the choice $\alpha(z)=\frac{1}{\gamma}z$, $\gamma>0$, leading to the class of the {\it multiplier risk measures}, with multiplier $\gamma$ being interpreted as a measure of the risk manager's ambiguity aversion. Other choices are of course possible.

By definition, a Fr\'echet risk measure depends on the choice of the metric $\d$. In the remainder of the paper, we focus on the use of Wasserstein metric for the metrization of ${\cal P}(\Omega)={\cal P}(\R^{d})$, however for the sake of comparison we consider some extensions in terms of popular pseudo-metrics such as the Kullback-Leibler divergence.


\subsection{ \textit{\textbf{Wasserstein Barycentric Risk Measures}}}

We now define a special class of Fr\'echet  risk measures, the Wasserstein Barycentric risk measures. This corresponds to the choice $\d(\mu_1,\mu_2) = W_p(\mu_1,\mu_2)$, $p>1$ with $W_p(\cdot,\cdot)$ being the p-Wasserstein distance defined as
\begin{eqnarray*}
\mathcal{W}_p(\mu_1, \mu_2) :=\left\{ \inf  \left(\int_{\R^d \times \R^d} |x - y|^p d\nu(x,y) \,\,\, : \,\,\, \nu \in\Pi(\mu_1,\mu_2)    \right) \right\}^{1/p},
\end{eqnarray*}
where $\Pi(\mu_1,\mu_2)$ denotes the set of all transport plans between $\mu_1$ and $\mu_2$, i.e.,  all measures on 
$\Omega\times \Omega=\R^d \times \R^d$ with marginals $\mu_1$ and $\mu_2$. Endowed with this metric the space of probability measures with finite $p$-moments, becomes a complete metric space which will be denoted by $\P_p(\R^d)$. It is well known (\cite{villani2003topics}) that this distance provides a metrization of the space of probability measures compatible with the weak*-topology, and recent research (\cite{py2016a}) has shown that the use of this distance allows for better aggregation of probability models as compared to other choices of metrics or pseudo-metrics. In the case where $\Omega=\R$, the solution to this problem can be expressed in terms of the generalized inverses of the distribution functions $F_i$ corresponding to the measures $\mu_i$, i.e. the quantiles $g_i$, and the Wasserstein distance admits the convenient integral expression 
\begin{equation}\label{Qrep}
\mathcal{W}_p(\mu_1, \mu_2) = \left\{\int_0^1 |g_1(s) - g_2(s)|^p ds \right\}^{1/p}.
\end{equation}

The special case where $p=1$, and the distribution is a discrete distribution on the real line is related to the Gini distance.
 For a detailed account of the properties and history of the Wasserstein distance see \cite{santambrogio2015optimal} and \cite{villani2003topics}. The Wasserstein barycenter has been well studied in the particular case where $p=2$ and its existence and uniqueness has been proved in the case where probability measures are absolutely continuous with respect to the Lebesgue measure on $\R^d$ (\cite{agueh2011barycenters}).  Furthermore, for the case $p=2$ closed form solutions for the Wasserstein barycenter can be obtained for important families of probability measures on $\R^{d}$, rendering the use of the Wasserstein metric very attractive. For the above reasons we will focus our attention to this particular case and only present a brief discussion of the general case $p \ne 2$.

For the special choice of the metric $\d(\cdot, \cdot)=W_2(\cdot, \dot)$ in Definition \ref{frisk} we obtain the definition of the Wasserstein barycentric risk measures.

\begin{definition}[\textbf{\textit{Wasserstein Barycentric Risk Measure}}]\label{wrm} \hfill

Given a risk $X : \Omega=\R^d \to \R$, a set of priors $\M=\{ \mu_1,..., \mu_n \}$, a set of weights $w=(w_1,\ldots, w_n) \in \Delta^{n-1}$, and a multiplier $\gamma\in(0,\infty)$, we define 
 the \textit{Wasserstein barycentric risk measure} $\RW$ by 
\begin{eqnarray}\label{rm-1}
\RW(X) := \sup_{\mu \in \P(\Omega)} \left\{  \mathbb{E}_\mu[-X] - \frac{1}{2 \gamma} F_{\M}(\mu) \right\},\nnb \\
F_{\M}(\mu) := \sum_{i=1}^n w_i W_2^2(\mu,\mu_i) - V_{\M}, \\
V_{\M}:= \inf_{\mu \in \P(\Omega)} \sum_{i=1}^n w_i W_2^2(\mu,\mu_i)\nnb
\end{eqnarray} 
The risk measure depends on the choice of  $\gamma$ and $w$ but the dependence is not made explicit in order to ease notation.  
\end{definition}

Following up on the discussion based on the results of Proposition \ref{Prop2.2}, the positive real number $\gamma >0$  in Definition \ref{wrm} 
 plays the role of a  multiplier that quantifies how much the variance of the priors in $\M$ is penalized in the determination of the risk measure for $X$. The smaller the value of  $\gamma$ is, the more severely is the variance of priors penalized, therefore, driving the minimizer of problem \eqref{rm-1} as close as possible to the  Fr\'echet mean of $\M$ (barycenter) as quantified in terms of the $2$-Wasserstein distance. In this sense, $\gamma$ can be thought of as a measure of the uncertainty aversion of the risk manager, since an uncertainty averse agent would favor a unique model, therefore likely being hostile to the existence of many and possibly conflicting models, which is quantified by large Fr\'echet variance. 

The variational problem \eqref{rm-1} is a well-posed variational problem, which defines a convex risk measure that enjoys the general properties of Fr\'echet risk measures presented in Proposition \ref{Prop2.2}.

\begin{remark}[{\bf Coherent version of Wasserstein barycentric risk measures}] \hfill

\noindent One can also consider the closely related class of coherent Wasserstein barycentric risk measures, which can be defined as 
\begin{eqnarray}\label{wbc}
\rho_{\eta }(X)=\sup_{\mu \in C_{\eta}  }{\mathbb E}_{\mu}[-X],
\end{eqnarray}
where $C_{\eta} \subset \P_{2}(\R^d)$ is the convex set
 $C_{\eta} := \{  \mu \in\P_2(\R^d)  :  \sum_{i=1}^n w_i W_2^2(\mu, \mu_i) \leq \eta \}$. These measures are coherent  according to the definition proposed in the pioneering work of \cite{artzner1999coherent}. Problem \eqref{wbc} can be reduced to problem \eqref{rm-1} using the Lagrange multiplier formulation. 
\end{remark}

\begin{remark}[Extension of Definition \ref{wrm}]\label{GENERALIZED-WRM} One of course may choose to metrize  ${\cal P}(\R^d)$ using other Wasserstein metrics that $W_{2}$, depending on the properties of the  probability measures considered. A possible extension of Defnition \ref{wrm} may be to define
\begin{eqnarray*}\label{rm-1-1}
\RW(X) := \sup_{\mu \in \P(\Omega)} \left\{  \mathbb{E}_\mu[-X] - \frac{1}{q \gamma} F_{\M}(\mu) \right\},
\end{eqnarray*} 
for the choice $F_{\M}(\mu) := \sum_{i=1}^n w_i W_p^q(\mu,\mu_i) - V_{\M}$ and $V_{\M}:= \inf_{\mu \in \P(\Omega)} \sum_{i=1}^n w_i W_p^q(\mu,\mu_i) $, and $p>1$, $q>1$.  Such choices may lead to well posed variational problems but have to be treated numerically. Furthermore, one may choose more general probability spaces $\Omega$, keeping of course in mind to check for the existence of the Wasserstein barycenter in this more general setting.
\end{remark}

We now present some results concerning the calculation of the convex risk measures $\RW$ in an almost closed form for some cases which often occur in real life applications.

\subsubsection{The case of a single random factor}

First, we consider the case where $d=1$ i.e. the case where the position $X$ is affected only through one risk factor $Z$. 

\begin{proposition}\label{Prop2.4}
Let $\M=\{\mu_1,\ldots,\mu_n\}$ be a set of alternative prior probability models. Let $\Omega=\R$, $-X = \Phi_0(Z)$ where $Z:\R\to\R$ is considered as the only random factor that affects the risk $X$,
and $\{g_{i}\}_{i=1}^n$ the corresponding quantile functions of the models for $Z$ in $\M$. The following results hold.

\noindent (a). The variational problem \eqref{rm-1}  is equivalent to the variational problem
\begin{eqnarray}\label{rep1}
\RW(X)=\max_{g \in \mathbb{S}} \left\{  \int_{0}^{1} \left( \Phi_0(g(s)) - \frac{1}{2\gamma}  (g(s)-g_{B}(s) )^{2} \right)  ds \right\}
\end{eqnarray}
where $\mathbb{S}$ is the space of quantile functions and  $g_B(s) := \sum_{i=1}^n w_i g_i(s)$,  which is well posed.

\noindent (b). If $\Phi_0$ is smooth enough, $\gamma$ is small enough and we consider continuous distributions for the priors the maximizer $g$ of \eqref{rep1} can be obtained in terms of the quantile function by solving the equation
\begin{equation}\label{1D-FOC}
\Lambda(g) := g(s) -\gamma \frac{d}{dz}\Phi_0( g(s) ) - g_B(s) = 0.
\end{equation}
\end{proposition}

\begin{proof}
For the proof of the Proposition please see Appendix \ref{A.2}.
\end{proof}

The following examples illustrate the use of the above proposition.

\begin{example}[{\bf Affine risk mapping}]

In the case where the risk mapping is of affine form, i.e. $\Phi(z) = \alpha + b z$, for $\alpha,b \in \R$,  an application of Proposition \ref{Prop2.4} yields  the closed form solution 
\begin{equation}\label{wassrm-l}
\RW(X) = \mathbb{E}_{\mu_B}[-X] + \frac{\gamma b^2}{2} ,
\end{equation}
where the unique maximizer is the probability model which can be represented in terms of its the quantile function $g:= g_{B} + \kappa = \sum_{i=1}^{n} w_{i} g_{i}+\kappa$ where $\kappa = b\gamma$. The maximizer in this case can be identified as the weighted quantile average (i.e. the barycenter) shifted by the multiplier $\kappa$ (left shift if $\kappa<0$ or right shift if $\kappa>0$. Note the distortion effect of the risk mapping on the maximizer and the value of risk measure.
\end{example}

\begin{example}[{\bf Quadratic risk mapping - $\Delta-\Gamma$ approximation }]
In the case where the risk mapping is of quadratic form, i.e. $\Phi(z) = \alpha + b z + c z^2$ for $\alpha, b, c \in \R$, corresponding to the popular $\Delta-\Gamma$ approximation of any risk position, an application of Proposition \ref{Prop2.4}
yields the closed form solution 
\begin{equation}\label{wassrm-q}
\RW(X) = \int_0^1 \Phi_0\left( \frac{g_B(s)+\kappa}{\lambda} \right) ds - \frac{1}{2\gamma} \int_0^1 \left( \frac{g_B(s) + \kappa}{\lambda} - g_B(s) \right)^2 ds,
\end{equation}
where the unique maximizer is the probability model represented in terms of the distorted, by the risk mapping, barycentric quantile function $g:=\frac{g_B + \kappa}{\lambda}=\frac{\sum_{i=1}^{n} w_{i} g_{i} +\kappa}{\lambda}$ where $\kappa = b\gamma$ and $\lambda = 1 - c\gamma$, which can be identified as the weighted quantile average shifted by the multiplier $\kappa$ and scaled by the multiplier $\lambda$. Note that in both cases, when $\gamma \to 0$ we recover the barycenter as the maximizer. 
\end{example}

\begin{remark}
For smooth $\Phi_0$ and $\gamma$ small enough, and in the case where $\M$ consists of continuous distributions one may obtain a perturbative expansion for the risk measure in terms of the expression
$$ \rho_W(X) = \int_0^1  \left( \Phi_0(g_B(s)) + \frac{\gamma}{2}(\Phi'_0(g_B(s))^2 \right) ds  + \mathcal{O}(\gamma^2).$$
\end{remark}

More general cases can be treated either by numerical treatment of the parametric algebraic equation \eqref{GENERALIZED-WRM}, or by direct treatment of the variational problem \eqref{rep1} by discretization or by inclusion of a penalization term  which guarantees that the solution is in ${\mathbb S}$.

\subsubsection{The case of multiple random factors: Location-Scatter family}

Now we consider the case that $d$ random factors affect the position $X$ through a risk mapping $\Phi_{0}(\cdot)$, i.e. $-X = \Phi_{0}(Z)$. We assume that the risk factors $Z=(Z_1,\ldots, Z_d)$ follow a Location - Scatter family so that $Z = m + S^{1/2} Z_0$, where $m \in \R^{d}$, $S \in {\mathbb P}(d) \subset \R^{d \times d}$ is a positive semidefinite matrix and $Z_{0}$ is a random variable on $\R^{d}$. We will use the notation ${\mathbb P}(d)$ for the set of positive semidefinite $d \times d$ matrices, and we will denote by $\nu$, the probability measure on $\R^{d}$, which corresponds to the probability distribution of the central random variable $Z_0$, which characterizes the family. We will denote such random variables by $Z \in LS(m,S)$, and slightly abusing notation we will denote the corresponding probability measure again as $\mu = LS(m,S)$.  Examples of such families are e.g. the normal family, the Student family or other families of elliptical distributions, which are widely used as models in risk management (see e.g \cite{mcneil2015quantitative}). Note that the existence of fat tails is not excluded by this choice of models, as the random variable $Z_0$ may be chosen so as to exhibit such behaviour if this is necessary.

For the case where $\M=\{\mu_i\}_{i=1}^n$, $\mu_i = LS(m_i, S_i)$, $(m_i,S_i)\in \R^d \times \mathbb{P}(d)$, the Fr\'echet function when the measure space is metrized in terms of the Wasserstein distance admits the form
\begin{eqnarray*}
{\mathbb F}_{\M}(m,S):=\overline{{\mathbb F}}_{\M}(m)+{\mathbb F}_{\M}^{0}(S):= \sum_{i=1}^{n}w_{i}\| m- m_{i} \|^2+ \sum_{i=1} w_{i} Tr\left( S + S_i -2 (S_i^{1/2} S S_{i}^{1/2})^{1/2} \right).
\end{eqnarray*}
It has been shown that the Wasserstein barycenter $\mu_{B}$ is also a measure $\mu_{B}=LS(m_{B},S_{B})$ with $(m_B,S_{B}) \in \R^{d} \times {\mathbb P}(d)$ being obtained as the solution of the matrix optimization problem
\begin{eqnarray*}
\min_{(m,S)  \in \R^{d} \times \R^{d \times d}} {\mathbb F}_{\M}(m,S)=\min_{(m,S)  \in \R^{d} \times \R^{d \times d}} \sum_{i=1}^{n}w_{i}\| m- m_{i} \|^2+ \sum_{i=1} w_{i} Tr\left(S + S_i -2 (S_i^{1/2} S S_{i}^{1/2})^{1/2} \right)
\end{eqnarray*}
It is easily seen  that $m_{B}=\sum_{i=1}^{n}w_{i} m_{i}$. Furthermore, it can be shown (see e.g. \cite{alvarez2016fixed} or \cite{bhatia2017bures}) that $S_{B}$ is the solution of the matrix equation 
\begin{eqnarray}\label{BAR-EQUATION}
S_{B}=\sum_{i=1}^{n}w_i (S_{B}^{1/2} S_{i}S_{B}^{1/2})^{1/2}.
\end{eqnarray}
The Fr\'echet variance breaks into two contributions $V_{\M}=\overline{V}_{\M}+V^{0}_{\M}$, where 
\begin{eqnarray*}
\overline{V}_{\M}=\overline{{\mathbb F}}_{\M}(m_{B})=\sum_{i=1}^{n} w_{i} \| m_i - m_{B}\|^{2},
\end{eqnarray*}
and $V^{0}_{\M}=\min_{S \in {\mathbb P}(d)} {\mathbb F}^{0}_{\M}(S)={\mathbb F}^{0}_{\M}(S_{B})$, through which we may  define the normalized Fr\'echet functions 
\begin{eqnarray*}
\overline{F}_{M}(m):=\overline{{\mathbb F}}_{\M}(m)-\overline{V}_{\M}, \,\, \mbox{and} \,\, F^{0}_{\M}(S):={\mathbb F}^{0}_{\M}(S)-V^{0}_{\M}.
\end{eqnarray*}

We now consider the problem of calculating the Wasserstein barycentric risk measure in the case where there is consent that the risk factors follow this general family, however there is uncertainty as to the parameters of the model, i.e., the set of priors $\M$ consists of probability measures $\mu_i \in {\cal P}(\R^d)$ such that $\mu_{i}=LS(m_i, S_{i})$, $m_i \in \R^{d}$, $S_i \in \R^{d \times d}$, $i=1,\ldots, n$. Equivalently, according to each of the priors in $\M$, the risk factor $Z \sim m_{i} + S_{i}^{1/2} Z_0$, where the $d$-dimensional random variable $Z_0$ is distributed by the probability measure $\nu \in {\cal P}(\R^{d})$.

We will need the following definitions.  Let $\Phi_0 : \R^{d} \to \R$ be a function, associated with the risk mapping of a particular position, assumed to be sufficiently smooth. Define the functions $\Phi, \Phi_{\ell}, \Psi_{\ell,k} : \R^{d} \times {\mathbb P}(d) \to \R$ by
\begin{eqnarray}\label{NEW-FUN-DEF}
\Phi(m,S):=\int_{\R^{d}} \Phi_0(m + S^{1/2} z) d\nu(z),  \nonumber\\
\Phi_{\ell}(m,S):=\int_{\R^{d}} D_{\ell}\Phi_0(m+ S^{1/2} z)   d\nu(z), \\
\Psi_{\ell, k}(m ,S):=\int_{\R^{d}} D_{\ell}\Phi_0(m+ S^{1/2} z) z_{k}  d\nu(z), \nonumber
\end{eqnarray}
and we will also use the notation
\begin{eqnarray}\label{NOTATION-0}
D_{m}\Phi(m,S)=[\Phi_{1},\ldots, \Phi_{d}]', \,\,\, D_{S}\Phi(m,S)= \frac{1}{2}\left( \Psi S^{-1/2} + S^{-1/2} \Psi^T \right).
\end{eqnarray}

\begin{proposition}\label{Prop2.8}
Assume that $\M=\{\mu_{i}, \,\,\, i=1,\ldots, n\}$ with $\mu_{i}=LS(m_{i},S_{i})$ with $m_{i} \in \R^{d}$ and $S_{i} \in {\mathbb P}(d)$, where ${\mathbb P}(d) \subset \R^{d \times d}$ is the set of positive definite and symmetric matrices. If the position of the firm is provided by the risk mapping $-X=\Phi_0(Z)$ then $\RW(X)$ is calculated as the solution of the  matrix optimization problem
\begin{eqnarray}\label{MATRIX-WAS-OPT}
\RW(X)=\sup_{(m,S) \in \R^{d} \times {\mathbb P}(d)} \left\{ \Phi(m,S)  -\frac{1}{2\gamma} \left(   \overline{F}_{\M}(m) +  F^{0}_{\M}(S) \right) \right\}
\end{eqnarray}
The maximizer $(m,S)$ to the above problem can be found as the solution of the set  of matrix equations (derived from the first order conditions):
\begin{equation}\label{FOC-A}
\begin{aligned}
\gamma D_{m}\Phi(m,S)-\left(m-\sum_{i=1}^{n}w_{i} m_{i}\right)=0 \\
2 \gamma S^{1/2} D_{S}\Phi(m,S) S^{1/2} - \left( S-\sum_{i=1}^{n} w_{i} (S^{1/2} S_{i} S^{1/2})^{1/2} \right)= 0.
\end{aligned}
\end{equation}
\end{proposition}

\begin{proof}
For the proof of the Proposition please see Appendix \ref{A.3}.
\end{proof}

The first order conditions \eqref{FOC-A} may not be solved analytically (even for the case where $\gamma=0$)  but may be approximated numerically with a fixed point scheme of the form
\begin{equation}\label{FIXED-POINT-SCHEME}
\begin{aligned}
&m^{(k+1)}=\sum_{i=1}^{n}w_{i} m_{i}+\gamma D_{m}\Phi(m^{(k)},S^{(k)}), \\
&S^{(k)}S^{(k+1)} =\left( \sum_{i=1}^{n} w_{i} ((S^{(k)})^{1/2} S_{i}( S^{(k)})^{1/2} )^{1/2} + 2 \gamma (S^{(k)})^{1/2} D_{S}\Phi(m^{(k)},S^{(k)})( S^{(k)})^{1/2}\right)^2
\end{aligned}
\end{equation}
This scheme in the case where $\gamma=0$ reduces to the fixed point scheme for the calculation of the Wasserstein barycenter the convergence of which was shown in \cite{alvarez2016fixed} or \cite{bhatia2017bures}). For small enough values of $\gamma$ this fixed point scheme can be treated as a perturbation of a converging fixed point scheme which assuming sufficient smoothness for the function $\Phi$,  can be shown to converge.

\begin{remark}[{\bf A perturbative approach to the calculation of $\RW(X)$}]\label{Rem2.9}

It can be noted that the second matrix equation in \eqref{FOC-A}  is a perturbation of the matrix equation for the determination of the Wasserstein barycenter in the case of location - scatter families. This  indicates that for small $\gamma$ the solution to this system of matrix equations will be concentrated around the Wasserstein barycenter  $(m_{B},S_{B})$. Introducing the notation
\begin{equation}\label{NOTATION-123}
\begin{aligned}
M_{B}=D_{m}\Phi(m_{B},S_{B}), \,\,\, C_{B}=D_{S}\Phi(m_{B},S_{B}), \,\,\, B=S_{B}^{1/2},\\
B_{i}=S_{i}^{1/2}, \,\,\, D_{i}=S_{i}S_{B}^{1/2}, \,\,\, E_{i}=(S_{B}^{1/2} S_{i} S_{B}^{1/2})^{1/2}, \,\,\, G_{i}=(S_{i}^{1/2} S_{B} S_{i}^{1/2})^{1/2},
\end{aligned}
\end{equation}
 it  can be shown (see Appendix \ref{A.4} for details) that for small values of $\gamma$ the risk measure admits the expansion
\begin{eqnarray}\label{EXPANSION-000}
\RW(X)= \Phi(m_{B},S_{B}) + \gamma \left(\frac{1}{2} \| M_{B} \|^2 +  Tr(C_{B} \tS) + \frac{1}{2} Tr\left( \sum_{i=1}^{n} w_{i} \Z_{i}\right) \right)
\end{eqnarray}
where $\tS$ is the part of the solution $(\tS, \J, \H_{1},\ldots, \H_{n})$ of the linear system of matrix equations
\begin{equation}\label{PPP-1}
\begin{aligned}
\tS- \sum_{i=1}^{n}w_{i} \H_{i} = 2 B W  B, \\
\tS- \J B - B \J =0, \\
\J D_{i} + D_{i} \J - \H_{i} E_{i} -E_{i} \H_{i} =0, \,\,\, i=1,\ldots, n,
\end{aligned}
\end{equation}
for $W=C_{B}$ and the matrices $(\Z_1,\Z_2, \ldots, \Z_{n})$ solve the linear system of (decoupled) Sylvester equations
\begin{eqnarray}\label{PPP-2}
\Z_{i} G_{i} + G_{i} \Z_{i} = - 2\Y_{i}^2, \,\,\, i=1, \ldots, n,
\end{eqnarray}
where the right hand sides are determined by the solution of the Sylvester equations
\begin{eqnarray}\label{PPP-3}
\Y_{i} G_{i} +  G_{i} \Y_{i} = B_{i}\tS  B_{i}, \,\,\, i=1,\ldots, n .
\end{eqnarray}
The maximizing measure has mean $m=m_{B}+\gamma \tm$, and covariance matrix $S=S_{B}+\gamma \tS$.
\end{remark}

We close the discussion by providing some examples of possible risk mappings and the application of Proposition \ref{Prop2.8}.

\begin{example}[{\bf Linear risk mappings}] Linear risk mappings are quite often used in quantitative risk management as approximations of more complicated nonlinear risk mappings. One great advantage of using the linear approximation is that it allows for closed form expressions for risk measures, and as a result of that risk management procedures which are based upon the linear approximation are often used in practice. 

Consider a linear risk mapping $-X=\langle a,Z \rangle$ where $a \in \R^{d}$ and $\langle \cdot, \cdot \rangle$ denotes the standard inner product in $\R^{d}$. Then we may calculate explicitly the function $\Phi$ as
\begin{eqnarray*}
\Phi(m,S)={\mathbb E}_{\nu}[(a,m +S^{1/2} Z)]=\langle a,m \rangle.
\end{eqnarray*}
For this risk mapping $D_{m}\Phi(m,S)=a$ and $D_{S}\Phi(m,S)=0$, so that  the system \eqref{FOC-A} becomes
\begin{equation*}\label{FOC-A}
\begin{aligned}
\gamma a-(m-\sum_{i=1}^{n}w_{i} m_{i})=0, \\
S-\sum_{i=1}^{n} w_{i} (S^{1/2} S_{i} S^{1/2} )^{1/2}=0,
\end{aligned}
\end{equation*}
which can be readily solved to yield $m=m_{B}+\gamma a$ and $S=S_{B}$. Sustituting the maximizer we obtain the risk measure as
\begin{eqnarray*}
\RW(X)=\langle a,m_{B}\rangle +\frac{ \gamma}{2} \| a \|^2.
\end{eqnarray*}
 Note that this result is exact (and not perturbative), and observe the distortion on the maximizing measure and its effect on the risk measure. Note further that the same result applies whenever the risk mapping is such that $\Psi_{\ell, k}(m,S)=0$, i.e. whenever the random variables $D_{\ell}\Phi_{0}(m+ S^{1/2} Z_0) Z_{0,k}$ are centered, a case that arises for specific symmetries of the risk mapping $\Phi_{0}$.

This result can be extended for the case of a portfolio of assets described by linear risk mappings. Assuming that each asset is described by the risk mapping $-X_{k}=\langle a_{k},Z \rangle$, $a_k \in \R^{d}$, $k=1,\ldots, K$, we have that  for a portfolio $\theta=(\theta_1,\ldots, \theta_K)$ the total position is $X=\sum_{k=1}^{K} \theta_{k} X_{k}$, so that the risk mapping is $\Phi_{0}(Z)=\langle \sum_{k=1}^{K}\theta_{k} a_{k}, Z \rangle$ and the function $\Phi(m,S)=\langle \sum_{k=1}^{K}\theta_{k} a_{k}, m\rangle$. Therefore, $D_{m}\Phi(m,S)=\sum_{k=1}^{K} a_{k}$ and $D_{S}\Phi(m,S)=0$, so that the maximizer is $m=m_{B}+\gamma ( \sum_{k=1}^{K}\theta_{k} a_{k} )$, $S=S_{B}$ and the risk measure is calculated as 
\begin{eqnarray*}
\RW(X)=\sum_{k=1}^{K} \theta_{k} \langle  a_{k}, m_{B}\rangle + \frac{\gamma}{2}  \| \sum_{k=1}^{K}\theta_{k} a_{k} \|^2.
\end{eqnarray*}
\end{example}

For more general risk mappings we will also have a distortion (as compared to the Wasserstein barycenter) not only for the mean but also for the covariance matrix. The following example illustrates that.

\begin{example}[{\bf Quadratic risk mapping}] The case of quadratic risk mappings is part of a very common approximation scheme in quantitative risk management, usually described as the $\Delta-\Gamma$ approximation. Consider a risk mapping of the form $-X=\langle a,Z \rangle + \langle Z,AZ \rangle$ where $a \in \R^d$ and $A \in \R^{d \times d}$, assumed symmetric without loss of generality. For this risk mapping we may calculate
\begin{eqnarray*}
\Phi(m,S) &=& {\mathbb E}[\langle a,m+S^{1/2} Z_0\rangle + \langle m+S^{1/2} Z_0,A m + A S^{1/2} Z_0\rangle ]\\
&=& \langle a,m\rangle+\langle m,Am\rangle +Tr(S^{1/2} A S^{1/2}).
\end{eqnarray*}
We may also calculate $[D_{m}\Phi(m,S)](\tm)=\langle a+2Am,\tm \rangle$, for every $\tm \in \R^d$ and $[D_{S}\Phi(m,S)](\tS)=Tr( A \tS )$  for every $\tS \in {\mathbb P}(d)$, where we used the cyclicity of the trace. Using Proposition \ref{Prop2.8} we see that the maximizer in this case can be obtained as the solution of the system of matrix equations
\begin{equation*}
\begin{aligned}
\gamma (a + 2 A m) -(m-m_{B})=0, \\
2 \gamma  S^{1/2} A S^{1/2} -\left(S-\sum_{i=1}^{n} w_{i} (S^{1/2} S_{i} S^{1/2} )^{1/2} \right)=0.
\end{aligned}
\end{equation*}
These equations are decoupled, the first one is readily solved to yield $m = (I - 2 \gamma A)^{-1}(m_{B} +\gamma a)$, whereas the second equation has either to be solved numerically using an iterative scheme or can be approximated in a perturbative manner using the method described in Remark \ref{Rem2.9}. In this case the maximizer measure has  mean $m=m_{B} +\gamma  \tm$ where $\tm= a +2 A m_{B}$ and
covariance $S=S_{B}+ \gamma \tS$ where the matrix $\tS$ is obtained upon solving for the unknown matrices $(\tS, J, H_1,\ldots, H_n)$  the system of linear equations  \eqref{PPP-1} for the choice $W=A)$. Once this system is solved (numerically in the general case) the risk measure is approximated to first order in $\gamma$ by
\begin{eqnarray*}
\RW(X)=(a,m_{B})+(m_{B},A m_{B}) + Tr ( S_{B}^{1/2} A S_{B}^{1/2})+ \gamma ( \Gamma_1 + \Gamma_2),
\end{eqnarray*}
where
\begin{eqnarray*}
\Gamma_1=\frac{1}{2} \| m_{B} + A m_{B} \|^2, \\
\Gamma_2 = Tr(A \tS) + \frac{1}{2} Tr\left(\sum_{i=1}^{n} w_{i} Z_{i}\right),
\end{eqnarray*}
where the matrices $Z_i$, $i=1,\ldots, n$ are solutions of the system of  decoupled Sylvester equations \eqref{PPP-2}. In the case of a  portfolio $\theta=(\theta_1,\ldots, \theta_K)$ of $K$ assets or risk positions $-X_{k}= \langle a_{k},Z\rangle + \langle Z,A_{k} Z\rangle$, $a_{k} \in \R^{d}$, $A_{k} \in \R^{d \times d}$, symmetric (without loss of generality), $k=1,\ldots, K$,  the risk mapping becomes 
$\Phi(m,S)=\sum_{k=1}^{K} \theta_{k} \left[ \langle a_{k}, m\rangle + \langle m, A_{i} m\rangle + Tr(S^{1/2} A_{k} S^{1/2})\right]$ and accordingly $[D_{m}\Phi(m,S)](\tm)=\sum_{k=1}^{K} \theta_{k}\langle a_{k}+2A_{k}m,\tm\rangle$, for every $\tm \in \R^d$ and 
$$[D_{S}\Phi(m,S)](\tS)=\sum_{k=1}^{K}\theta_{k} Tr(A_{k} \tS)$$  
for every $\tS \in {\mathbb P}(d)$.  The correction to the barycenter is obtained by the solution of system \eqref{PPP-1} for the choice $W=\sum_{k=1}^{K} \theta_{k} A_{k}$ .
\end{example}

\subsection{ \textit{\textbf{The Weighted Entropic Risk Measure}} }

Relative entropy (or the Kullback-Leibler divergence), while not a true metric in the space of probability measures, has long been one of the favourite measures of divergence between probability models. Apart from its use in information theory it has played a crucial role in the study of model uncertainty in economic theory (see e.g. \cite{hansen2008robustness}) and has also been used in the study of risk management through the definition of the so called entropic risk measures for the case of single priors (see \cite{follmer2011entropic}; see also \cite{papayiannis2016numerical}). On account of the popularity of entropy and in order to generalize the entropic risk measures first proposed in \cite{follmer2011entropic} for the single prior case to the multi-prior case we extend our definition of Fr\'echet risk measures to include pseudo-metrics rather than just metrics, and  choose $\d$ to be the Kullback-Leibler divergence ($\mathcal{KL}$). This choice leads to a particular class of Fr\'echet risk measures which is called the {\it weighted entropic risk measure}, and can be a considered as a generalization to the multi-prior setting of the class of entropic risk measures (see e.g. \cite{frittelli2004dynamic, barrieu2007pricing, follmer2011entropic, ahmadi2012entropic}). For simplicity we will restrict our attention for  such risk measures to the case where the set of priors $\M$ consists of probability measures which are absolutely continuous with respect to the Lebesgue measure on $\Omega$. 
We recall that in this case if $f_{i}$ are the probability densities of the measures $\mu_i$ then ${\cal KL}(\mu_1 \| \mu_2)=\int_{\Omega} f_1(\omega) \log \frac{f_1(\omega)}{f_2(\omega)} dm$. 

The definition of the {\it weighted entropic risk measure} follows.

\begin{definition}[\textbf{The weighted entropic risk measure}] \hfill

\noindent Given a risk $X : \Omega \to \R$, a set of priors $\M=\{ \mu_1,..., \mu_n \}$ (absolutely continuous with respect to the Lebesgue measure $m$), a set of weights $w=(w_1,\ldots, w_n) \in \Delta^{n-1}$, and a multiplier $\gamma\in(0,\infty)$,
 the weighted entropic risk measure admits a robust representation and is defined by\footnote{We use $1/\gamma$ in the definition of \eqref{wentropic} instead of $1/2\gamma$ used in the definition of general Fr\'echet risk measures, so as to be in line with the standard definition of the entropic risk measure.}
\begin{eqnarray}\label{wentropic}
\RE(X) := \sup_{\mu\in\P(\Omega)} \left\{ \mathbb{E}_\mu[-X] - \frac{1}{\gamma} \mathcal{KL}_{w,\M}(\mu) \right\},
\end{eqnarray}
where $\mathcal{KL}_{w,\M}(\mu) := \sum_{i=1}^n w_i \mathcal{KL}(\mu \| \mu_i) - V_{\M}$ and $V_{\M} := \min_{\mu \in \P(\Omega)} \mathcal{KL}_{w,\M}(\mu)$.
\end{definition}

The next proposition provides  explicit results  concerning the calculation of $\RE$.

\begin{proposition}\label{Prop2.15} 
Assume that the risk $X$ is provided through the risk mapping $-X=\Phi_0(Z)$ where $Z$ is a vector of $d$ random factors that affect the risk position. Let $\M=\{ \mu_1,..., \mu_n \}$ be a set of priors for $Z$ (absolutely continuous with respect to the Lebesgue measure $m$), $w=(w_1,\ldots, w_n) \in \Delta^{n-1}$ a set of weights, and $\gamma \in (0,\infty)$. Then, the weighted entropic risk measure  $\RE$ defined in \eqref{wentropic} has value
\begin{eqnarray}\label{wvalue}
\RE(X) = \frac{1}{\gamma} \log\left( \int e^{-\gamma X} f_0 \ud m \right) = \frac{1}{\gamma} \log\left( \int_{\R^d} e^{\gamma \Phi_0(z)} f_{0}(z)  \ud z \right),
\end{eqnarray}
where $f_{0} := C_0 f_G = C_0 \prod_{i=1}^n f_i^{w_i} $ and $C_0 := \left( \int f_G dm \right)^{-1}$ the normalizing constant. Moreover, the optimal probability measure, $\mu_{\gamma}$, is fully characterized in terms of its probability density function $f_{\gamma} := C_{\gamma} e^{-\gamma X} f_G$, where $C_{ \gamma} := \left( \int e^{-\gamma X} f_{G} \ud m \right)^{-1},$ denotes the normalizing constant.
\end{proposition}

\begin{proof}
For the proof of the Proposition see Appendix \ref{A.5}. 
\end{proof}

\begin{remark}
Although $\RE$ is the natural extension of the entropic risk measure to the case where more than one priors are available, there are some drawbacks caused by the more complex framework. The main disadvantage is that, by construction, the weighted KL-divergence requires all the probability models in the set of priors to have common support otherwise, the KL-barycenter density function is a degenerate distribution. Therefore, this risk measure is advised for cases where the priors are close to each other. The weighted KL-divergence performs better when the priors display a high level of homogeneity. Some evidence regarding the use of this distance functional in practice within the context of a learning scheme, and assessment of its performance as compared to other distance functionals are provided in \cite{py2016a} . 
\end{remark}


\section{ Applications of Fr\'echet Risk Measures in Insurance }

The proposed Fr\'echet risk measures can be used in any financial and insurance application in which there is uncertainty in form of multiple models for the risk factors which affect the total position of the firm. We choose two characteristic examples to illustrate and motivate their use in insurance risk management, (a) their use in risk allocation of the firm and (b) their use in premia calculation. Both examples have been chosen as real world applications in which model uncertainty appears, and in which the choice of a wrong probability model may have important negative consequences for the firm.

\subsection{ \textit{\textbf{Risk Allocation under Model Uncertainty }}}

Risk allocation among the different sectors of an insurance firm is one of the major problems in the insurance business since the ability of a firm to fulfill efficiently this task affects its performance and longevity. Consider for example a firm which is composed by $K$ different sectors, and assign to each one, a random variable $X_i$ which represents the profit/loss of each particular sector $k$ for all $k=1,2,...,K$.

In general, the financial behaviour of each sector (and the whole firm) is estimated and modeled according to certain factors (the so called {\it risk factors}), monitored in the market and the economy, which affect the risk. Obviously, an efficient estimation of the random behaviour of these factors should lead to more efficient estimation of the risk of a firm and its particular sectors, since these risks are directly affected from the risk factors stature through the relevant risk mappings. However, when the firm needs to estimate its total risk by a risk measure, let us say $\rho(\cdot)$, then the firm actually needs to know or to have properly modeled the random behavior of the risk factors $Z$ that affect its risk (this interaction is introduced by a risk mapping, e.g. $X = -\Phi_0(Z)$), by a probability measure $\mu \in \P(\R^d)$ (in the case of $d$ random factors). Note that the risk factors that affect the risk of a firm are common for all of its sectors, however each of the sectors is affected in a different way and this is introduced through its particular risk mapping, i.e. the position for the sector $k$ is defined as $X_k = -\Phi_{0,k}(Z)$ where $\Phi_{0,k} : \R^d \to \R$ is its own risk mapping (i.e. how this particular sector quantifies different states of the market/economy). However, the main problem in such operations is that the probability measure $\mu \in \P(\R^d)$ describing the joint random behaviour of risk factors $Z$ is rarely sufficiently estimated even in the case that the marginal behaviour of each risk factor $Z_i$ is sufficiently modeled. The problem of modeling dependence and its effects on capital risk capital allocation for the firm is a very important problem with deep theoretical and serious practical implications which have been extensively studied in the insurance literature (see e.g. \cite{de2012modeling}, \cite{bernard2014risk}, \cite{bernard2016measuring},  \cite{jakobsons2016general}, \cite{liu2017collective}).

In this case, the risk manager of the firm asks for the consultancy by a number of experts (let us say $n$) i.e. asks for their opinion regarding the joint distribution of the random vector $Z=(Z_1,...,Z_d)$ and as a result is provided with $n$ prior probability models $\mu_i$  $i=1,\ldots, n$ where each one quantifies the partial information that is available to each expert. These partial models must be aggregated to a single model for $Z$, which will be used for the quantification of the sector risks and the total risk of the firm and for capital allocation considerations. We maintain that the Fr\'echet risk measures proposed in this work is an ideal tool for such considerations. We focus in the case of the Wasserstein barycentric risk measures and provide expressions in closed or semi-closed form for the total risk and the risk allocation according to Euler's allocation principle, which can be used by practitioners for the efficient allocation of risk.

We start by recalling some fundamental facts concerning the Euler allocation principle for a general risk measure (see \cite{tasche2007capital}). According to the Euler allocation principle if $\rho(X)$ denotes the total risk of the firm then, the risk contribution of $X_{k}$, denoted by $\rho(X_{k} \,|\, X)$, can be computed as
\begin{equation}\label{risk-xi}
\rho(X_{k} \,|\, X) = \left .\frac{d \rho}{d h}(X + h X_{k}) \right\vert_{h=0}.
\end{equation}
The Euler allocation principle satisfies a desirable axiomatic framework for capital allocation and can be used for the creation of quantitative indices such as Return on Risk Adjusted Capital (RORAC) or the Diversification Indices (DI) that can be used for the efficient risk management of the firm. 

In what follows, we perform a detailed solution of the capital asset allocation problem using Wasserstein barycentric risk measures working under the hypothesis that the risk factors can be modeled using the elliptic family of distributions. In particular, we consider the problem of risk allocation in the small $\gamma$ limit.  Assume that the risk factors $Z=(Z_1,\ldots, Z_{d})$  are modeled by a location scale family with the priors $\mu_{i}=LS(m_{i},S_{i})$ and that the firm has $K$ different lines of business, each with a risk mapping $-X_{k}=\Phi_{0,k}(Z)$, giving rise to the respective risk mappings $\Phi_{k}(m,S)$, $k=1,\ldots, K$, when the model $LS(m,S)$ is adopted for the risk factors.  The total position of the firm is  given by $-X=-\sum_{k=1}^{K} X_{k}$, so that the corresponding risk mapping is $\Phi(m,S)=\sum_{k=1}^{K} \Phi_{k}(m,S)$. If the firm alters its position by leveraging slightly the $j$-th line of business then the total position is given by $-X(\epsilon)=-X -\epsilon X_{j}$ and the corresponding risk mapping will be $\Phi(m,S; \epsilon) =\Phi(m,S) + \epsilon \Phi_{j}(m,S)$. To apply the Euler risk allocation principle we need to calculate $\RW(X(\epsilon))$ and then calculate $\RW(X_{j} \mid X):=\left .\frac{d}{d\epsilon} \RW(X(\epsilon)) \right\vert_{\epsilon =0}$. In general this calculation can be done numerically using Proposition \ref{Prop2.8} for the calculation of $\RW(X)$ and $\RW(X(\epsilon))$ and then calculate $\RW(X_j \mid X)$ using numerical differentiation. However, in the limit of small $\gamma$, the perturbative expansion outlined in Remark \ref{Rem2.9} can be used to obtain a result in almost closed form, subject to the numerical solution of a linear matrix equation. We will  also use  the notation of \eqref{NOTATION-123}, as well  as the following simplified notation: 
\begin{equation}\label{NOTATION-SIM}
\begin{aligned}
A_{k}=\Phi_{k}(m_{B},S_{B}),  \,\,\, M_{k}=D_{m}\Phi_{k}(m_{B},S_{B}),  \,\,\,C_{k}=D_{S}\Phi_{k}(m_{B},S_{B}), \,\,\, \\
 A=\sum_{k=1}^{K} A_{k}, \,\,\,  C= \sum_{k=1}^{K} C_{k}, \,\,\,   M=\sum_{k=1}^{K} M_{k}.
\end{aligned}
\end{equation}

Let $\tS, \Y_{i}, \Z_{i}$, $i=1,\ldots, n$,  be the solution of the system of linear matrix equations  \eqref{PPP-1}-\eqref{PPP-3} for the choice $W=C$, and consider the system
\begin{equation}\label{RDP-2-0}
\begin{aligned}
\tS'-\sum_{i=1}^{n}w_{i} \H'_{i} = 2 B  C_{j}  B, \\
\tS'-\J' B - B \J'=0, \\
\J' D_{i} + D_{i} \J' - \H'_i E_{i} -E_{i} \H'_{i} =0,\\
\Y'_{i}G_{i} + G_{i} \Y'_{i} = B_{i} \tS' B_{i}, \,\,\, i=1,\ldots, n,\\
\Z'_{i}G_{i} + G_{i} \Z'_{i} = -2 \Y'_{i} \Y_i-2 \Y'_{i} \Y_{i}, \,\,\, i=1,\ldots, n.
\end{aligned}
\end{equation}

\begin{proposition}\label{Prop3.1}
The risk measure of the position $X=\sum_{k=1}^{K} X_{k}$, up to first order in $\gamma$ is provided by the expression
\begin{equation}\label{EXPANSION-000-EPS-0}
\begin{aligned}
\RW(X)= 
A  + \gamma \left(  \frac{1}{2} \| M  \|^2 + Tr( C \tS) +\frac{1}{2} Tr( \sum_{i=1}^{n} w_{i}\Z_{i}) \right),
\end{aligned}
\end{equation}
The corresponding risk allocation is
\begin{eqnarray*}
\RW(X_{k} \mid X) = A_{j} + \gamma\left( \langle M,M_{j}\rangle+ Tr(C_{j} \tS) + \frac{1}{2} Tr(\sum_{i=1}^{n} w_{i} \Z'_{i}) \right).
\end{eqnarray*}
where the matrices $\tS$, $\Z_{i}$, $\Z_i'$ are obtained  by the solution of the  linear matrix equations \eqref{PPP-1} (for the choice $W=C$), \eqref{PPP-2}   and \eqref{RDP-2-0}.
\end{proposition}

\begin{proof}
For the proof of the Proposition see Appendix \ref{A.7}.
\end{proof}

\subsection{\textit{\textbf{Risk Premia Estimation under Model Uncertainty}}}

A standard problem in insurance is the estimation of the total claim amount that an insurance firm is obliged to cover at a specific horizon $T>0$ (assumed fixed). The number of the claim events that occur within an interval $[0,T]$ is described by the random variable $N$ while each one of the claim sizes are described by the random variables $C_i$ for $i=1,2,\ldots,N$. Such a calculation provides an estimate of the liabilities of the firm, that should be covered by the insurance premia obtained from the insured customers. A crucial calculation for the viability of the insurance firm is a robust estimation of the total amount of liabilities, so that the insurance premia can be calculated in such a way that the probability of ruin or insolvency of the firm is minimized. Standard premia calculations involve the use of a risk measure (see e.g. \cite{mcneil2015quantitative}) in order to quantify the risk of the firm as a result of its liabilities, which are thereafter distributed to the customers according to their reliability or their needs.

A standard model used in quantifying the total liabilities of the firm is the compound  mixed Poisson process (see e.g. \cite{mikosch2009non}), according to which the total claims amount by time $T$ is given by $-X= \sum_{j=1}^{N} C_j$, where $N$ is assumed to be a Poisson process of possibly stochastic rate $\lambda_N$ and $\{C_j\}_{j=1,...,N}$ are assumed to be independent and identically distributed according to a probability law $F$, with $C_j$ and $N$ independent. We assume that for the fixed horizon $T>0$, considered to be the horizon of operation of the firm, we wish to quantify the total risk incurred by the firm on account of the liabilities $-X$. For that we need the probability distribution for the random variable $X$ and a risk measure $\rho(X)$  of the insurer's choice in order to calculate the risk of the position of the firm. The total amount of premia collected by the firm must then be equal to $\pi(X) = ( 1+\alpha) \rho( X ) $ where $\alpha>0$ is a safety loading factor. Certain popular risk measures used in premium calculations are the variance risk measure, leading to the classical standard deviation premium principle or an exponential utility function leading to exponential premium principle.

In most cases however, there is an important component of model uncertainty in the above calculations, and often there are diverging opinions and models as to the distribution of the number of claims or the size of the claims. It is the aim of this section to assess the applicability of the Fr\'echet risk measures proposed in this paper to the problem of robust premium calculation under model uncertainty. One of the main characteristics of the risk measures proposed in this work, is that they have the ability of filtering possible diverging information concerning the random variable in question, thus leading to an aggregate robust probability model for the characterization of the random variable in question.

We consider that uncertainty concerning the distribution of $X$ is introduced in terms of a set of stochastic factors affecting the frequency of the claims as well as their severity. In particular, in order to a guarantee the independence between $N$ and $C_{i}$ we consider two distinct sets of stochastic factors $\Zob=(\Zo_{1},\ldots, \Zo_{d_1})$ and $\Ztb=(\Zt_{1},\ldots, \Zt_{d_2})$, independent and that there exists two mappings $\Phi^{(1)}_0 : \R^{d_1} \to \R_{+}$ and $\Phi^{(2)}_0 : \R^{d_2} \to \R_{+}$ such that 
$ N \sim Poisson( \Phi^{(1)}_0(\Zob))$, and $C_{i} = \Phi^{(2)}_0(\Ztb)$. Note that in general, $N$ may fail to be Poisson, however, due  to the independence of $N$ and $C_{i}$, it still holds that ${\mathbb E}[-X]={\mathbb E}[N] {\mathbb E}[C_{i}]$.
As a result of that (upon conditioning on $\Zob$, see e.g. \cite{rolski2009stochastic}, Section 4.3.3)  we have that 
\begin{eqnarray*}
{\mathbb E}[-X]= {\mathbb E}[ \Phi^{(1)}_0(\Zob)] \, {\mathbb E}[\Phi^{(2)}_0(\Ztb)],
\end{eqnarray*}
which allows us to construct the risk mapping in terms of the stochastic factors $Z=(\Zob,\Ztb)$.

We now assume a set of priors $\M = \M_1 \cup \M_2$ for the stochastic factors, where $\M_1 =\{\mu^{(1)}_{1}, \ldots, \mu^{(1)}_{n}\}$, $\M_2 =\{\mu^{(2)}_{1}, \ldots, \mu^{(2)}_{n}\}$, are priors for the stochastic factors $\Zob$, $\Ztb$ respectively and we make the assumption that each of these sets consists of location - scatter families of the form $\mu^{(1)}_{i} \sim LS(m^{(1)}_{i}, S^{(1)}_{i})$ and $\mu^{(2)}_{i} \sim LS(m^{(2)}_{i}, S^{(2)}_{i})$, for $i=1,\ldots, n$, where $(m^{(j)}_i, S^{(j)}_{i}) \in \R^{d_{j}} \times {\mathbb P}(d_{j})$, $j=1,2$, $i=1,\ldots, n$. Letting $\Zob_{0} \sim \nu^{(1)}$, and $\Ztb_{0} \sim \nu^{(2)}$ be the central random variables of the families of stochastic factors $\Zob$ and $\Ztb$ we may  define the following function $\Phi : \R^{d_1} \times {\mathbb  P}(d_1) \times \R^{d_2} \times {\mathbb  P}(d_2) \to \R$ which may be factored as the product of two functions $\Phi_{j} : \R^{d_j} \times {\mathbb  P}(d_j)  \to \R$, $j=1,2$ as $\Phi(m^{(1)},S^{(1)},m^{(2)},S^{(2)})=\Phi_{1}(m^{(1)},S^{(1)})  \Phi_{2}(m^{(2)},S^{(2)})$ where 
\begin{eqnarray*}
\Phi_{j}(\mu^{(j)},S^{(j)})=\int_{\R^{d_j}} \Phi^{(j)}_0 ( \mu^{(j)} + (S^{(j)})^{1/2} z)  d\nu^{(j)}(z), \,\,\, j=1,2.
\end{eqnarray*}

Upon defining the Fr\'echet functions
\begin{eqnarray*}
{\mathbb F}_{\M_j}(m^{(j)},S^{(j)}):= \sum_{i=1}^{n}w_{i}\| m^{(j)}- m^{(j)}_{i} \|^2\\
+ \sum_{i=1} w_{i} Tr(S^{(j)} + S^{(j)}_i -2 ((S^{(j)}_i)^{1/2} S^{(j)} (S^{(j)}_{i})^{1/2})^{1/2}, \,\,\, j=1,2,
\end{eqnarray*}
 the corresponding Fr\'echet variances $V_{\M_j}$,  and the normalized Fr\'echet functions $F_{\M_{j}}$, $j=1,2$,  for the two sets of priors  we can define the Wasserstein barycenter risk measure $\RW(X)$ in terms of the matrix optimization problem
\begin{eqnarray*}
\RW(X)=\mathop{\max_{(\mu^{(1)},S^{(1)}) \in \R^{d_1} \times {\mathbb P}(d_1)}}_{(\mu^{(2)},S^{(2)}) \in \R^{d_2} \times {\mathbb P}(d_2) }\large\{ \Phi_{1}(m^{(1)},S^{(1)})  \Phi_{2}(m^{(2)},S^{(2)}) \\
-\frac{1}{2\gamma} \left(F_{\M_{1}}(m^{(1)},S^{(1)}) + F_{\M_{2}}(m^{(2)},S^{(2)})\right) \large\}
\end{eqnarray*}
Following Proposition \ref{Prop2.8} we see that the maximizer is the solution of the system of matrix equations
\begin{equation}\label{FOC-A-PREM}
\begin{aligned}
\gamma D_{m^{(1)}}\Phi_{1}(m^{(1)},S^{(1)})\Phi_{2} (m^{(2)},S^{(2)}) -(m^{(1)}-\sum_{i=1}^{n}w_{i} m^{(1)}_{i})=0, \\
2 \gamma( S^{(1)})^{1/2} D_{S^{(1)}}\Phi(m^{(1)},S^{(1)})  \Phi_{2}(m^{(2)},S^{(2)})(S^{(1)})^{1/2} \\ - (S^{(1)}-\sum_{i=1}^{n} w_{i} ((S^{(1)})^{1/2} (S^{(1)}_{i}) ( S^{(1)})^{1/2} )^{1/2})=0,\\
\gamma\Phi_{1} (m^{(1)},S^{(1)}) D_{m^{(2)}}\Phi_{2}(m^{(2)},S^{(2)}) -(m^{(2)}-\sum_{i=1}^{n}w_{i} m^{(2)}_{i})=0, \\
2 \gamma( S^{(2)})^{1/2} \Phi_{1}(m^{(1)},S^{(1)})  D_{S^{(2)}}\Phi(m^{(2)},S^{(1)})  (S^{(2)})^{1/2} \\- (S^{(2)}-\sum_{i=1}^{n} w_{i} ((S^{(2)})^{1/2} (S^{(2)}_{i}) ( S^{(2)})^{1/2} )^{1/2})=0,
\end{aligned}
\end{equation}
In general system \eqref{FOC-A-PREM} must be solved numerically by an iterative scheme, or using a perturbative scheme 
similar to that proposed in Remark \ref{Rem2.9}. For the sake of an example providing an illustration of the use of Fr\'echet risk measures in risk premia calculations let us consider the case of a linear risk mapping of the form 
\begin{eqnarray*}
\Phi_0^{(j)}(Z^{(j)})=\langle a^{(j)}, Z^{(j)} \rangle, \,\,\, a^{(j)} \in \R^{d_j}, \,\,\, j=1,2.
\end{eqnarray*}
Then, system \eqref{FOC-A-PREM} is conveniently decoupled  to yield  the solution $(m^{(1)},S^{(1)}_{B}, m^{(2)}, S^{(2)}_{B})$ where $S^{(j)}_{B}$ is the covariance matrix Wasserstein barycenter of $\M_{j}$, $j=1,2$, and $(m^{(1)},m^{(2)})$  is the solution to the system
\begin{eqnarray}\label{222-1}
\gamma \langle a^{(2)}, m^{(2)} \rangle \,  a^{(1)}  - (m^{(1)}  - m^{(1)}_{B} )=0, \\
\gamma \langle a^{(1)}, m^{(1)} \rangle \,  a^{(2)}  - (m^{(2)}  - m^{(2)}_{B} )=0,
\end{eqnarray}
where $m^{(j)}_{B}=\sum_{i=1}^{n} w_{i} m^{(j)}_{i}$, $j=1,2$. The solution to this system will provide the minimizer, and from that the risk measure can be obtained, and hence a robust premium calculation principle can be  constructed.

For the sake of ilustration let us consider the case where $d_1=d_2=1$, in which $m^{(j)}=m_{j} \in \R$, and $S^{(j)}=\sigma_{j} \in \R_{+}$, $j=1,2$.  In this case we also have that $a^{(j)} =a_{j} \in \R$, $j=1,2$. In this case system \eqref{222-1} reduces to the linear system
\begin{eqnarray*}
\gamma a_{1} a_{2} m_{2} + m_{1,B} = m_{1}, \\
\gamma a_{1} a_{2} m_{1} + m_{2,B}= m_{2},
\end{eqnarray*}
where 
\begin{eqnarray*}
m_{j,B}=\sum_{i=1}^{n} w_{i} m_{j,i}, \,\,\, j=1,2,
\end{eqnarray*}
is the mean for the barycenter of $\M_{j}$, $j=1,2$,
which can be readily solved to yield
\begin{eqnarray*}
m_1= \left( \frac{1+ \gamma^2 a_1^2 a_2^2}{1-\gamma^2 a_1^2 a_2^2} \right) m_{1,B} + \left( \frac{\gamma a_1 a_2}{1-\gamma^2 a_1^2 a_2^2} \right) m_{2,B}, \\
m_2= \frac{1}{1-\gamma^2 a_1^2 a_2^2}m_{2,B} + \left( \frac{ \gamma a_1 a_2 }{1-\gamma^2 a_1^2 a_2^2} \right)  m_{1,B},
\end{eqnarray*}
while $\sigma_j$  is given by
\begin{eqnarray*}
\sigma_{j}=\left( \sum_{i=1}^{n} w_{i} \sigma_{j,i}^{1/2} \right)^2, \,\,\, j=1,2.
\end{eqnarray*}

In order to assess the ability of the proposed risk measures to filter out the uncertainty arising from multiple information sources, and construct an optimal aggregate model out of them, we construct the following thought experiment. There is a true probability model for the random factors $Z = (Z^{(1)}, Z^{(2)})$ (where each factor is considered as one dimensional for the sake of simplicity) which is known to the designer of the experiment but assumed to be unknown to the risk manager of the firm. The risk manager has access to a number of possibly diverging probability models for $Z = (Z^{(1)}, Z^{(2)})$, comprising a set of priors $\M$, which are constructed by the designer of the experiment as random perturbations of the true probability law. Using one risk measure of the Fr\'echet class proposed here, the risk manager can aggregate the information provided by the alternative models in $\M$ into a single model for $Z$ which will be used in order to quantify the risk. Since the true probability law is known to the experimenter, the success of each risk measure can be assessed by quantifying the deviation of the calculated risk from the true one. Clearly, the risk measure which provides results closer to the true value has done a better job in aggregating the diverging prior information therefore leading to a robust approximation of the risk. 

For the needs of the experiment presented here, we have assumed that $N$ follows the Poisson distribution with rate parameter $Z^{(1)}$ where $Z^{(1)}\sim N(m_1,s_1)$ and $C_j = Z^{(2)}$ where $Z^{(2)}\sim N(m_2,s_2)$, with the parameters chosen so that positivity of $N$ and $C_j$ is guaranteed. Other choices for the claim size's distribution have been considered (e.g. Gamma distribution or more random factor to affect the claim size) however they are not reported here for the sake of brevity. Model uncertainty has been introduced as uncertainty in the parameters $(m_1,s_1)$ and $(m_2,s_2)$ where by random perturbations of the true value a number of alternative scenarios for $Z$ have been created, collected in $\M$. Three different types of perturbation schemes have been employed depending on the size of the random perturbation. In the first type we have assumed small random perturbations around the true value of the parameters $(m_1,s_1,m_2,s_2)$, which leads to a set of priors with high homogeneity (or equivalently of low Fr\'echet variance), in the second type the perturbation was of medium size leading to a set of priors of medium homogeneity as measured by the Fr\'echet variance, whereas the third type consisted of large perturbations for the parameters leading to a prior set of low homogeneity. These three different perturbation protocols will be hereafter referred to as the ``high" (hh), ``medium" (mh) and ``low" (lh) homogeneity scenario respectively. 

\begin{table}[htp]
	\tiny
	\begin{center} 
		\begin{tabular}{|c|ccc|ccc|ccc|}
			\hline		
			$\bm \gamma$ &	
			\multicolumn{3}{c|}{ $\mathbb{E}[ \rho_{\gamma}(X) ]$}& 
			\multicolumn{3}{c|}{ $\mathbb{E}\left( \frac{| \rho_{\gamma}(X) - \rho_0(X) |}{\rho_0(X)} \right) $ }&
			\multicolumn{3}{c|}{$Var^{1/2} \left( \frac{| \rho_{\gamma}(X) - \rho_0(X) |}{\rho_0(X)} \right) $}\\ 
			&\textbf{\textit{Average}} & \textbf{\textit{Entropic}} & \textbf{\textit{Wasserstein}} & \textbf{\textit{Average}} & \textbf{\textit{Entropic}} & \textbf{\textit{Wasserstein}} & \textbf{\textit{Average}} &  \textbf{\textit{Entropic}} & \textbf{\textit{Wasserstein}} \\ \hline
			
			\multicolumn{10}{|c|}{\bf (hh) high homogeneity}\\  \hline
			\multicolumn{10}{|c|}{\bf i. n=5}\\ \hline 
			
			
			0.1000&	 5659.20&	11633.00&	  5679.70&	0.1318&	1.3267&	0.1359&	0.0165&	0.0689&	0.0143\\
			0.0500&	 5658.90&	10485.00&	  5257.20&	0.1318&	1.0969&	0.0514&	0.0163&	0.0641&	0.0136\\
		     0.0100&	 5660.10&	  8782.00&	  5126.90&	0.1320&	0.7564&	0.0255&	0.016	&  0.0554&	0.0123\\
			0.0050&	 5655.70&	  6816.00&	  5060.40&	0.1311&	0.3632&	0.0141&	0.0161&	0.0332&	0.0098\\
			0.0010&	 5663.40&	  5783.00&	  5014.90&	0.1327&	0.1567&	0.0103&	0.0163&	0.0172&	0.0080\\
			0.0000&	 5660.30&	  5658.00&	  5001.80&	0.1321&	0.1315&	0.0101&	0.0167&	0.0170&  0.0077\\

			\hline
			\multicolumn{10}{|c|}{\bf ii. n=10}\\ \hline 
			
			0.1000&	5659.20& 	    -    &  5682.10&	 0.1240&	          -& 0.1364&	0.0129&	     -   &	0.0100\\
              0.0500&	5619.90&	 10742.00&	5259.00&	 0.1230&  1.1484&	0.0518&	0.0124&	0.0537&	0.0092\\
			0.0100&	5615.20&	   8879.00&	5127.40&	 0.1232&	 0.7758&	0.0255&	0.0131&	0.0437&	0.0091\\
			0.0050&	5615.80&	   6795.00&	5064.00&	 0.1235&	 0.3590&	0.0134&	0.0132&	0.0239&	0.0082\\
			0.0010&	5618.30&	   5739.00&	5011.80& 0.1234&	 0.1479&	0.0073&	0.0121&	0.0124&	0.0055\\
			0.0000&	5617.10&	   5612.00&	5001.10&	 0.1232&	 0.1224&	0.0072&	0.0131&	0.0133&	0.0053\\
	
			\hline
			\multicolumn{10}{|c|}{\bf iii. n=30}\\ \hline 
			
			0.1000&	5560.10& 	        -&	5681.30&	 0.1120&		-   &   0.1363&	0.0095&	        -&   0.0060\\
			0.0500&	5558.20&	 11080.00&	5257.80&	 0.1116&	1.2161&	0.0516&	0.0092&	0.0417&	0.0051\\
			0.0100&	5559.50&	  8999.00&	5126.10&	 0.1119&	0.7997&	0.0252&	0.0091&	0.0285&	0.0053\\
			0.0050&	5556.40&	  6763.00&	5062.60&	 0.1113&	0.3525&	0.0126&	0.0094&	0.0143&	0.0052\\
			0.0010&	5561.20&	  5688.00&	5013.10&	 0.1122&	0.1376&	0.0047&	0.0089&	0.0087&	0.0034\\
			0.0000&	5557.70&	  5553.00&	5001.50& 0.1115&	0.1105&	0.0042&	0.0092&	0.0093&	0.0031\\
			
			\hline 
			\multicolumn{10}{|c|}{\bf (mh) medium homogeneity}\\ \hline

			\multicolumn{10}{|c|}{\bf i. n=5 }\\ \hline 
			
			0.1000&	5559.90& 	        -&	5684.60&	 0.1120&	         -&  0.1369&	0.0315&	       - &	0.0286\\
			0.0500&	5554.10&	 10930.00&	5259.70&	 0.1108&	 1.1860&	 0.0523&	0.0317&	0.2292&	0.0253\\
			0.0100&	5557.80&	   8523.00&	5132.10&	 0.1116&	 0.7045&	 0.0299&	0.0317&	0.2016&	0.0208\\
			0.0050&	5537.30&	   6491.00&	5055.70&	 0.1075&	 0.2981&	 0.0217&	0.0321&	0.0865&	0.0169\\
			0.0010&	5557.30&	   5599.00&	5020.90&	 0.1115&	 0.1198&	 0.0203&	0.0316&	0.0371&	0.0153\\
			0.0000&	5547.50&	   5471.00&	5004.60& 0.1095&	 0.0943&	 0.0192&	0.0307&	0.0338&	0.0148\\
			
			\hline
			\multicolumn{10}{|c|}{\bf ii. n=10}\\ \hline 
			
			0.1000&	5476.70&		        -&   5682.20&  0.0953&          -&		0.1363&	0.0095&		   -&   0.1364\\
			0.0500&	5478.40&	 11283.00&	5258.80&  0.0957&	  1.2161&	0.0516&	0.0092&	1.2565&	0.0518\\
			0.0100&	5476.10&	   8374.00&	5126.50&  0.0952&	  0.7997&	0.0252&	0.0091&	0.6749&	0.0267\\
			0.0050&	5478.80&	   6409.00&	5062.40&	  0.0958&  0.3525&	0.0126&	0.0094&	0.2818&	0.0180\\
			0.0010&	5472.70&	   5498.00&	5009.70&	  0.0945&  0.1376&	0.0047&	0.0089&	0.0997&	0.0144\\
			0.0000&	5471.50&	   5376.00&	5001.90&	  0.0943&  0.1105&	0.0042&	0.0092&	0.0753&	0.0144\\

			\hline
			\multicolumn{10}{|c|}{\bf iii. n=30}\\ \hline 
			
			0.1000&	5388.40& 	        -&	5678.90&  0.0777&		        -&  0.1358&	0.0156&	        -&	0.0117\\
			0.0500&	5388.00&	 11698.00&	5256.30&	  0.0776&	1.3396&	0.0513&	0.0152&	0.2215&	0.0108\\
			0.0100&	5388.90&	   8201.00&	5126.60&	  0.0778&	0.6402&	0.0254&	0.0149&	0.1351&	0.0101\\
			0.0050&	5390.00&	   6271.00&	5064.70&	  0.0780&	0.2542&	0.0140&	0.0156&	0.0487&	0.0090\\
			0.0010&	5389.80&	   5411.00&	5011.70&	  0.0780&	0.0823&	0.0087&	0.0151&	0.0198&	0.0062\\
			0.0000&	5386.80&	   5283.00&	5001.30&	  0.0774&	0.0569&	0.0086&	0.0153&	0.0180&	0.0063\\
			
			\hline 
			\multicolumn{10}{|c|}{\bf (lh) low homogeneity}\\  \hline
			\multicolumn{10}{|c|}{\bf i. n=5}\\ \hline

			0.1000&	5392.30& 	         -&	 5675.40&	0.0840&	        -&	0.1355&	0.0556&	        -&	0.0645\\
			0.0500&	5394.90&	  10862.00&	 5263.30&	0.0856&	1.1727&	0.0653&	0.0554&	0.4293&	0.0483\\
			0.0100&	5389.20&	    8020.00&	 5120.50&	0.0841&	0.6048&	0.0500&	0.0519&	0.2883&	0.0366\\
			0.0050&	5394.30&	    6256.00&	 5064.10&	0.0841&	0.2533&	0.0484&	0.0571&	0.1303&	0.0382\\
			0.0010&	5405.50&	    5395.00&	 5015.70&	0.0870&	0.0933&	0.0448&	0.0561&	0.0626&	0.0355\\
			0.0000&	5390.00&	    5228.00&	 4998.00&	0.0839&	0.0724&	0.0459&	0.0549&	0.0538&	0.0345\\
			
			\hline
			\multicolumn{10}{|c|}{\bf ii. n=10}\\ \hline 
			
			0.1000&	5285.40& 	        -&	5676.40&  0.0607&		        -&  0.1353&	0.0397&	        -&	0.0451\\
			0.0500&	5278.50&	 10782.00&	5245.40&	  0.0600&	1.1566&	0.0545&	0.0392&	0.4551&	0.0365\\
			0.0100&	5292.60&	  7507.00&	5126.00&	  0.0619&	0.5020&	0.0386&	0.0395&	0.2388&	0.0286\\
			0.0050&	5297.70&	  6038.00&	5067.50&	  0.0636&	0.2105&	0.0346&	0.0386&	0.0972&	0.0251\\
			0.0010&	5290.30&	  5250.00&	5009.50&	  0.0621&	0.0699&	0.0333&	0.0394&	0.0508&	0.0244\\
			0.0000&	5302.00&	  5132.00&	5009.30&	  0.0634&	0.0532&	0.0318&	0.0395&	0.0420&	0.0238\\
			
			\hline
			\multicolumn{10}{|c|}{\bf iii. n=30}\\ \hline 
			
			0.1000&	5195.20& 	         -&	 5682.40&	0.0404&		    -&  0.1365&	0.0230&	        -&	0.0256\\
			0.0500&	5193.00&	  10296.00&	 5260.50&	0.0404&	1.0596&	0.0524&	0.0233&	0.4497&	0.0240\\
			0.0100&	5185.40&	    6937.00&	 5124.70&	0.0384&	0.3884&	0.0281&	0.0229&	0.1668&	0.0195\\
			0.0050&	5196.20&	    5861.00&	 5065.40&	0.0408&	0.1753&	0.0216&	0.0239&	0.0738&	0.0163\\
			0.0010&	5189.70&	    5172.00&	 5009.10&	0.0395&	0.0534&	0.0184&	0.0222&	0.0390&	0.0139\\
			0.0000&	5189.70&	    5047.00&  4998.00&	0.0394&	0.0393&	0.0180&	0.0231&	0.0361&	0.0135\\
			
			\hline
		\end{tabular}
		\caption{ Testing of the robust properties for the proposed Fr\'echet risk measures in the calculation of the total risk of the liabilities position $X$ under different homogeneity scenarios within the prior set. }
		\label{tab2}
	\end{center}
\end{table} 

Since for each homogeneity scenario a multitude of models regarding the random behaviour of $Z$ are available, and having in mind the discussion in Section \ref{Section2}, the risk manager should aggregate appropriately the prior information in the set $\M$, employing a Fr\'echet risk measure type, i.e. solving the related optimization problem 
\begin{eqnarray}\label{frm}
\max_{\mu\in\mathcal{P}_{\F}} \left\{ \mathbb{E}_{\mu}[-X] - \frac{1}{2\gamma} F_{\M}(\mu)  \right\}   
\end{eqnarray}
where $\mathcal{P}_{\F}$  is the admissible set of probability models indicated by the Fr\'echet function choice, and for a variety of choices for the uncertainty aversion parameter $\gamma$. In that case, the maximizer $\mu^{*}=\mu_{\gamma}$ of the problem \eqref{frm} is the optimal aggregation probability model, with respect to the Fr\'echet function chosen, which better represents the aggregate belief regarding the random behaviour of $Z$ taking into account the risk manager's preferences\footnote{The risk manager's preferences are introduced through the choice of the parameter $\gamma$. If $\gamma$ is chosen such that $\gamma = 0$ then the barycenter model under the notion of the respective Fr\'echet variance is selected. The parameter $\gamma$ can be interpreted as the risk manager's confidence to the provided prior information.}. Moreover, the quantity $\mathbb{E}_{\mu^{*}}[-X]$, represents the estimate of the total claim amount calculated under the optimal aggregation model $\mu^{*}$. Obviously, different estimates for the total claim amount shall be obtained under different choices of Fr\'echet functions and as a result using different Fr\'echet risk measures. Thus, it is clear that the robustness of the risk measure depends on the Fr\'echet function (or Fr\'echet variance) choice. In order to clarify this effect, we perform the following simulation study.

We test the aggregation models obtained by the two different Fr\'echet variance types discussed in this work, i.e. the Kullback-Leibler divergence barycenter (WKL) and the quadratic Wasserstein Barycenter (QWB), for their robustness properties under (a) several scenarios of homogeneity (high, medium and low), (b) different choices of the parameter $\gamma$ (corresponding to different levels of tolerance of the risk manager towards the deviance of the measure from the Fr\'echet barycenter; larger values allow for larger deviance from the mean element) and (c) different number of expert opinions ($n=5,10$ and $30$). For comparison reasons, we also include the simple average risk measure, i.e. the risk measure which treats the information provided from the prior set $\M$ by averaging the provided probability models. In order to assess the behavior of each Fr\'echet risk measure, we simulate each scenario for $B=1000$ times and then we evaluate the mean behavior of the risk measure by some appropriate quantities, namely
\begin{itemize}
\item the expected value of the risk measure, $\mathbb{E}[ \rho_{\gamma}(X) ]$, where expectation is taken over all different simulations,
\item the expected (with respect to the different simulations) relative difference of the risk measure value with respect to the true one, $\mathbb{E}\left( \frac{| \rho_{\gamma}(X) - \rho_0(X) |}{\rho_0(X)} \right)$ and 
\item the standard deviation (with respect to the different simulations) of the relative difference of the risk measure with respect the true value, $Var^{1/2} \left( \frac{| \rho_{\gamma}(X) - \rho_0(X) |}{\rho_0(X)} \right)$.
\end{itemize}
We recall that $\mu_{\gamma}$ and $\rho_{\gamma}(X)$ is the aggregate probability measure and risk measure respectively and $\mu_0$ and $\rho_0(X)$ is the true probability measure for $Z$ and the true risk measure value respectively,  known to the designer of the experiment. Note that $\mu_{\gamma}$ and $\rho_{\gamma}(X)$ depend on the choice of Fr\'echet function used, i.e. will differ if the Wasserstein metric is employed instead of the weighted KL-Divergence. 

The results are illustrated in Table \ref{tab2}. As it is shown from the evidence in Table \ref{tab2}, the Wasserstein barycentric risk measure appears to be consistently more robust in filtering out the diverging prior information than the other two choices even in the cases where the number of expert opinions is small ($n=5$). The homogeneity level within the prior set (i.e. the different expert opinions) has a smaller effect on the Wasserstein barycentric risk measure and the evidence of the simulation study indicates that its estimates converge to the real ones faster than the other choices in comparison, when the number of expert opinions $n$ grows. On the other hand, the weighted entropic risk measure displays worse performance that the Wasserstein risk measure, especially in the case where the homogeneity of the priors is low and seems to be more sensitive to the choice of $\gamma$. However, its behaviour seems to improve when $\gamma$ is taken small so that we get closer to the original barycenter. The average risk measure also performs worse that the Wasserstein risk measure, except of course for choices of $\gamma$ large, in which the large value of the uncertainty aversion parameter invalidates the action of the Fr\'echet penalty function, having as a result the inability of the risk measure in filtering out the uncertainty. 

Therefore, we can conclude that the Wasserstein barycentric risk measure demonstrates the most robust behavior as compared to the other risk measures proposed in this work. The results displayed support our original claim, that the use of the Wasserstein barycentric risk measure for the premium calculation will therefore lead to more accurate results hence, implying smaller losses for the firm and less financial burden for the customers therefore having better risk transfer properties. We do not find the superior performance of the Wasserstein barycentric risk measure with respect to robustness strange, since it is based on a true metric on the space of probability measures thus better quantifying the distance between the various priors and the Fr\'echet variance of the prior set, over a wide range of possible distribution families possibly diverging from the Normal distribution.

\section{Conclusions}

In this work we proposed a novel class of multi-prior convex risk measures, the class of Fr\'echet risk measures which is well-suited for quantifying the risk under uncertainty, in cases where multiple diverging models concerning the true distribution of the risk. These risk measures are designed so as to filter uncertainty, by aggregating available models in terms of a penalty function of Fr\'echet type. This results in choosing an aggregate model for the risk as close as possible to all available models, which in some sense minimizes model uncertainty as quantified by the Fr\'echet variance on the set of the available probability models. Geometrically this can be interpreted as choosing a probability model which is close to the barycenter of the set of priors in order to describe and represent the risk. The proposed risk measures are tested in two characteristic problems from insurance, (a) the risk capital allocation for an insurance firm and (b) the problem of premium calculation of an insurance firm.

%
%

\appendix 



\section{Proofs}

\subsection{Proof of Proposition \ref{Prop2.2}}\label{A.1}
\begin{proof} (i) By the definition of $\RF$ it holds that for any $\mu \in {\cal P}(\R^{d})$,  $\RF(X) \ge {\mathbb E}_{\mu}[\Phi_0(Z)]-\frac{1}{2\gamma} \alpha(\mu)$, so that choosing $\mu=\mu_{B}$ and keeping in mind that $a(0)=0$ we conclude  that
for any $\gamma >0$ it holds that $\RF(X ;\gamma) \ge {\mathbb E}_{\mu_{B}}[\Phi_0(Z)] \ge 0$. Furthermore by the convexity of $F_{\M}$ and the fact that $a$ is increasing we conclude that  $ a :=  \frac{1}{2\gamma}\alpha \,\circ \, F_{\mu} : {\mathcal P}(\R^d) \to \R_{+}$ is a convex function. Therefore,  by the definition of $\RF$, we have that $\RF(X)=\sup_{\mu \in {\cal P}(\R^{d})}[ {\mathbb E}_{\mu}[-X] - a(\mu) ]$, hence it is a convex risk measure by the robust representation theorem of \cite{follmer2002robust}.

\noindent (ii) Let $0\le \gamma_1 \le \gamma_2$.  For any $\mu \in {\cal P}(\R^{d})$ we have that
\begin{eqnarray*}
{\mathbb E}_{\mu}[-X] -\frac{1}{2\gamma_1} \alpha(F_{\M}(\mu) \le {\mathbb E}_{\mu}[-X] -\frac{1}{2\gamma_2} \alpha(F_{\M}(\mu) 
 \le \sup_{\mu \in {\cal P}(\R^d)} \left\{ {\mathbb E}_{\mu}[-X] -\frac{1}{2\gamma_2} \alpha(F_{\M}(\mu) \right\}\\=\RF(X ;\gamma_2),
\end{eqnarray*}
and taking the supremum over all $\mu \in {\cal P}(\R^d)$ on the left hand side we conclude that $\RF(X ;\gamma_1) \le \RF(X ; \gamma_2)$.

Since $\alpha(F_{\M}(\mu)) \ge 0$ for every $\mu \in {\cal P}(\R^d)$, we have that $-\frac{1}{2\gamma} \alpha(F_{\M}(\mu)) \to -\infty$  in the limit as $\gamma \to 0$, except for the choice $\mu=\mu_{B}$ for which this term vanishes. This fact, combined with the monotonicity with respect to $\gamma$, leads to the result that $\lim_{\gamma \to 0^{+}} \RF(X ;\gamma)={\mathbb E}_{\mu_{B}}[-X]$. On the other hand in the limit as $\gamma \to \infty$ the penalty term is inactive, so that $\lim_{\gamma \to \infty} \RF(X ; \gamma)=\sup_{\mu \in {\cal P}(\R^d)} {\mathbb E}_{\mu}[-X]$. The supremum is attained on the probability measure which is a Dirac measure on the $Z$ for which the essential supremum of $-X=\Phi_0(Z)$ is attained.
\end{proof}


\subsection{Proof of Proposition \ref{Prop2.4}}\label{A.2}

\begin{proof}
(a). We express the Fr\'echet function as well as the term ${\mathbb E}_{\mu}[-X]$ (using the risk mapping $-X=\Phi_0(Z)$), in terms of the quantiles of the measures in $\M$ and the maximizer (using the representation of the Wasserstein distance in terms of the quantile functions in the one dimensional case stated in \eqref{Qrep}), $g_i$ and $g$ respectively (note that $g_i, g$ are  the quantiles for random variable $Z$). Then it is straightforward to check that the problem can be represented as
\begin{eqnarray*} \max_{g\in\mathbb{S}} \left\{ \int_0^1 \Phi_0(g(s))ds - \frac{1}{2\gamma}\left(\sum_{i=1}^n w_i\int_0^1|g(s)-g_i(s)|^2 ds - V_{\M}\right) \right\}\\
 = \max_{g \in \mathbb{S}} \left\{  \int_{0}^{1} \left( \Phi_0(g(s)) - \frac{1}{2\gamma}  (g(s)-g_{B}(s) )^{2} \right)  ds \right\}
 \end{eqnarray*}
where by ${\mathbb S}$ we denote the space of quantile functions and the second expression is derived by straightforward algebraic manipulation using the definition of $V_{\M}=\inf_{g\in\mathbb{S}} \sum_{i=1}^n w_i \int_0^1 (g(s) - g_i(s))^2 ds = \mathbb{F}_{\M}(\mu_B)$ with the latter term expressed in terms of quantiles. This leads, after a change of sign, to the minimization problem $\inf_{g \in {\mathbb S}} \int_{0}^{1} U(s,g(s)) ds$, where $U(s,g(s))=-\Phi_0( g(s) )+ \frac{1}{2\gamma}( \sum_{i=1}^{n}w_{i} |g(s)-g_{i}(s)|^{2} - V_{\M})$. The existence of a minimizer for the relaxed problem can be obtained by a standard application of the direct method of the calculus of variations, by constructing a minimizing sequence  $\{g_{n}\} \in {\mathbb S}$  and using weak compactness results to guarantee the existence of a $g \in L^{2}([0,1])$ such that that $g_{n} \rightharpoonup g$ in $L^{2}([0,1])$. Indeed, a minimizing sequence is norm bounded, a fact that guarantees the existence of a weakly convergent subsequence. Then, since $U(s,\cdot)$ is convex  it follows by Theorem 6.54  in \cite{fonseca2007modern} that the functional $g \mapsto  \int_{0}^{1} U(s,g(s)) ds$ is weakly lower semi-continuous in $L^{2}([0,1])$ and the existence is guaranteed in $L^{2}([0,1])$. It remains to check that the minimizer is indeed a quantile function, i.e. that it is increasing and right continuous. By an application of Mazur's lemma there exists a  new sequence $\{ \tilde{g}_{n}\}$, the terms of which are convex combinations of the minimizing sequence $\{g_{n}\}$  such that $\tilde{g}_{n} \to g$ in $L^{2}([0,1])$, with the convergence being strong. Since the terms of $\{\tilde{g}_{n}\}$ are convex combinations of the terms of the minimizing sequence $\{g_{n}\}$ and $\{g_{n}\} \subset  {\mathbb S}$ it follows that $\{\tilde{g}_{n}\} \subset {\mathbb S}$. Since $\tilde{g}_{n} \to g$ in $L^{2}([0,1])$ (strong), there exists a subsequence (not renamed) such that $\tilde{g}_{n} \to g$ a.e. in $[0,1]$, so that $g$ is increasing. A further application of Egorov's theorem allows us to pass to a further subsequence converging uniformly to the same limit, from which right continuity follows. Hence the minimizer $g \in {\mathbb S}$. Uniqueness on ${\mathbb S}$ follows from the strict convexity of the functional $g \mapsto  \int_{0}^{1} U(s,g(s)) ds$. The above arguments readily generalize for the treatment of the generalized Wasserstein risk measures discussed in Remark \ref{GENERALIZED-WRM},  for $d=1$. \\

\noindent (b) The minimum Wasserstein distance is attained by the Wasserstein barycenter $\mu_{B}\in\P(\R)$, which in terms of the quantile functions is expressed as $g_B(s) = \sum_{i=1}^n w_i g_i(s)$ and the corresponding minimum value is $V_{\M} = \sum_{i=1}^n w_i \int_{0}^{1} (g_B(s) - g_i(s) )^2 ds$. Plugging this expression in relation \eqref{rep1}, and after a few calculations, the corresponding variational problem takes the form 
$$\sup_{g\in\mathbb{S}} \left\{ \int_0^1 \Phi_0(g(s))ds - \frac{1}{2\gamma} \int_0^1(g(s)-g_B(s) )^2 ds \right\},$$
where the maximum is attained as we showed in (a). In case where $\Phi_0(\cdot)$ is smooth enough (i.e. differentiable), the maximizer of the problem can be obtained in a semi-analytic form in terms of the quantile function. Indeed, taking first order conditions in the above expression, we obtain the maximizer $g$ as the solution of the variational inequality
$$ \int_0^1 \{ \Phi_0'(g(s)) - \frac{1}{\gamma}(g(s)-g_B(s))\}\tilde{g}(s) ds \leq 0, $$
for every $\tilde{g}$ such that $g + \epsilon \tilde{g} \in \mathbb{S}$ for small enough $\epsilon$. For $\gamma$ small enough and if $\M$ consists of continuous distributions we are allowed to consider $\tilde{g}$ and $-\tilde{g}$ in the above variational inequality leading to a first order condition of the form
 $$ \Lambda(g) = g(s) - \gamma\frac{d}{dz}\Phi_0(g(s)) - g_B(s) = 0.$$
\end{proof}


\subsection{Proof of Proposition \ref{Prop2.8}}\label{A.3}

Before the proof of Proposition \ref{Prop2.8} we need to state and prove the following lemmas.

\begin{lemma}\label{FR-FUN-DER} The following hold:

\begin{itemize}

\item[(i)]  The normalized Fr\'echet function $\overline{F}_{\M} : \R^{d} \to \R_{+}$ defined by 
$\overline{F}_{\M}(m):=\sum_{i=1}^{n}w_{i} \|m -m_i\|^2 - \overline{V}_{\M}$ is twice differentiable on $\R^{d}$ and satisfies
\begin{eqnarray*}
[D\overline{F}_{\M}(m)](\tm):= \left . \frac{d }{d\epsilon} \overline{F}_{\M}(m + \epsilon \tm) \right\vert_{\epsilon=0} = 2 \langle m-m_{B}, \tm\rangle, \,\,\, \forall \, \tm \in \R^{d},
\end{eqnarray*}
where $\langle \cdot , \cdot \rangle$ denotes the usual inner product in $\R^{d}$, and
\begin{eqnarray*}
[D^2\overline{F}_{\M}(m)](\tm_1,\tm_2):= \left. \frac{\partial^2 }{\partial\epsilon_1\partial \epsilon_2} \overline{F}_{\M}(m + \epsilon_1 \tm_1 + \epsilon_2 \tm_2) \right\vert_{\epsilon_1=\epsilon_2=0} = 2 \langle\tm_1,\tm_2\rangle, \,\,\, \forall \, \tm_1,\tm_2 \in \R^{d}.
\end{eqnarray*}
On the Wasserstein barycenter $m_{B}$ it holds that $\overline{F}_{\M}(m_{B})=0$ and $[D\overline{F}_{\M}(m_{B})](\tm)=0$ for every $\tm \in \R^d$.

\item[(ii)] The normalized Fr\'echet function  $F^{0}_{\M} : {\mathbb P}(d) \to \R_{+}$ defined by
 $$F_{\M}^{0}(S):=\sum_{i=1}^{n} w_{i} Tr\left( S + S_i -2 (S_i^{1/2} S S_{i}^{1/2})^{1/2} \right)-V^{0}_{\M}$$
 is twice differentiable on ${\mathbb P}(d)$ and satisfies
\begin{eqnarray*}
[DF^{0}_{\M}(S)](\tS)=\left . \frac{d}{d\epsilon}F^{0}_{\M}(S+\epsilon \tS) \right\vert_{\epsilon=0} = Tr\left( \tS - 2 \sum_{i=1}^{n} w_{i} Y_{i}(\tS) \right), \,\,\, \forall \, \tS \in {\mathbb P}(d),
\end{eqnarray*}
where $Y_{i} :=Y_{i}(\tS)$ is the solution of the Sylvester equation 
\begin{eqnarray}\label{SYLV-1}
Y_{i} (S_{i}^{1/2} S S_{i}^{1/2})^{1/2} +  (S_{i}^{1/2} S S_{i}^{1/2})^{1/2} Y_{i} = S_{i}^{1/2} \tS S_{i}^{1/2}.
\end{eqnarray}
and
\begin{eqnarray*}
[D^{2}F^{0}_{\M}(S)](\tS_1,\tS_2)= \left . \frac{\partial^2}{\partial \epsilon_1 \partial \epsilon_2}F^{0}_{\M}(S + \epsilon_1 \tS_1+\epsilon_2 \tS_2) \right\vert_{\epsilon_1=\epsilon_2 =0} \\
= -2 Tr( \sum_{i=1}^{n} w_{i} Z_{i}(\tS_1,\tS_2)), \,\,\, \forall \, \tS_1,\tS_2 \in {\mathbb P}(d),
\end{eqnarray*}
where $Z_{i}:=Z_{i}(\tS_1,\tS_2)$ is the solution of the Sylvester equation 
\begin{eqnarray}\label{SYLV-2}
Z_{i} (S_{i}^{1/2} S S_{i}^{1/2})^{1/2} +  (S_{i}^{1/2} S S_{i}^{1/2})^{1/2} Z_{i} = - Y_{i}(\tS_1)Y_{i}(\tS_2)-Y_{i}(\tS_2)Y_{i}(\tS_1).
\end{eqnarray}
On the Wasserstein barycenter $S_{B}$ it holds that $F^{0}_{\M}(S_{B})=0$ and $[DF^{0}_{\M}(S_{B})](\tS)=0$ for every $\tS \in {\mathbb P}(d)$. 
\end{itemize}

\end{lemma}

\begin{proof} (i) We can easily see that 
\begin{eqnarray*}
\overline{F}_{\M}(m+ \epsilon_1 \tm_1 + \epsilon_2 \tm_2)=\overline{F}_{\M}(m) + 2 \sum_{i=1}^{n} w_{i} \langle m-m_i,\epsilon_1 \tm_1 + \epsilon_2 \tm_2\rangle \\+ \sum_{i=1}^{n} w_{i} \langle\epsilon_1 \tm_1 + \epsilon_2 \tm_2, \epsilon_1 \tm_1 + \epsilon_2 \tm_2\rangle,
\end{eqnarray*}
from which the claim follows easily upon differentiating with respect to $\epsilon_1$ and $\epsilon_2$ and keeping in mind that $\sum_{i=1}^{n}w_{i}=1$.

(ii) We express $F^{0}_{\M}$ as $F^{0}_{\M}(S)=Tr(S -2 \sum_{i=1}^{n}w_{i}\phi_{i}(S)) + Tr(\sum_{i=1}^{n} w_{i} S_{i}) -V_{\M}$, where $\phi_{i}(S)=(S_{i}^{1/2}S S_{i}^{1/2})^{1/2}$.

For the function $\phi_{i} : {\mathbb P}(d) \to {\mathbb P}(d)$ it holds that
\begin{eqnarray}\label{LEM-031}
[D\phi_{i}(S)](\tS):=\left .\frac{d}{d\epsilon} \phi_{i}(S+\epsilon \tS) \right\vert_{\epsilon=0} = Y_{i},
\end{eqnarray}
where $Y_{i}:=Y_{i}(\tS)$ is the solution of the matrix Sylvester equation \eqref{SYLV-1} and
\begin{eqnarray}\label{LEM-032}
 [D^2\phi_{i}(S)](\tS_1,\tS_2):=\left .\frac{d^2}{d\epsilon^2} \phi_{i}(S+\epsilon \tS) \right\vert_{\epsilon=0} = Z_{i}(\tS_1,\tS_2),
\end{eqnarray}
where $Z_{i}:=Z_{i}(\tS_1,\tS_2)$ is the solution of the matrix Sylvester equation \eqref{SYLV-2}.
To calculate $[D\phi_{i}(S)](\tS)$ note that $\phi_{i}(S+\epsilon \tS) \phi_{i}(S+\epsilon \tS)=S_{i}^{1/2}(S+ \epsilon \tS ) S_{i}^{1/2}$,
differentiate with respect to $\epsilon$ to obtain
$$ \frac{d}{d \epsilon}\phi_{i}(S+\epsilon \tS) \phi_{i}(S + \epsilon \tS) + \phi_{i}(S + \epsilon \tS) \frac{d}{d\epsilon}\phi_{i}(S+\epsilon \tS) = S_{i}^{1/2}\tS S_{i}^{1/2}, $$
and then set $\epsilon=0$  to see that $Y_{i}(\tS)$ solves \eqref{SYLV-1}.

To calculate the second derivative $[D^{2}\phi_{i}(S)](\tS_1,\tS_2)$, note 
 that $\phi_{i}(S+\epsilon_1 \tS_1+ \epsilon_2 \tS_2) \phi_{i}(S+\epsilon_1 \tS_1 + \epsilon _2 \tS_2)=S_{i}^{1/2}(S+ \epsilon_1 \tS_1 +\epsilon_2 \tS_2) S_{i}^{1/2}$ and differentiating with respect to $\epsilon_1$ we obtain
\begin{eqnarray}\label{DIFF-01}
\frac{\partial}{\partial \epsilon_1}\phi_{i}(S+\epsilon_1 \tS_1 + \epsilon_2 \tS_2) \phi_{i}(S + \epsilon_1 \tS_1 + \epsilon_2 \tS_2) \nonumber  \\+ \phi_{i}(S + \epsilon_1 \tS_1 + \epsilon_2 \tS_2) \frac{\partial}{\partial\epsilon_1}\phi_{i}(S+\epsilon_1 \tS_1 +\epsilon_2 \tS_2) = S_{i}^{1/2}\tS_1, S_{i}^{1/2},
\end{eqnarray}
and setting $\epsilon_1=\epsilon_2=0$ we see that $\left . \frac{\partial}{\partial \epsilon_1}\phi_{i}(S+\epsilon_1 \tS_1 + \epsilon_2 \tS_2) \right \vert_{\epsilon_1=\epsilon_2=0}$ equals $Y_{i}(\tS_1)$ the solution of \eqref{SYLV-1} with $\tS=\tS_1$ on the right hand side.
We  further differentiate  \eqref{DIFF-01} once more with respect to $\epsilon_2$ to obtain
\begin{eqnarray*}
\frac{\partial^2}{\partial\epsilon_1 \partial \epsilon_2}\phi_{i}(S+\epsilon_1 \tS_1 +\epsilon_2 \tS_2)  \phi_{i}(S + \epsilon_1 \tS_1+\epsilon_2 \tS_2) \\
+   \frac{\partial}{\partial\epsilon_1}\phi_{i}(S+\epsilon_1 \tS_1+\epsilon_2 \tS_2)  \frac{\partial}{\partial\epsilon_2}\phi_{i}(S+\epsilon_1 \tS_1+ \epsilon_2 \tS_2) \\
+  \frac{\partial}{\partial\epsilon_2}\phi_{i}(S+\epsilon_1 \tS_1+\epsilon_2 \tS_2)  \frac{\partial}{\partial\epsilon_1}\phi_{i}(S+\epsilon_1 \tS_1+ \epsilon_2 \tS_2)\\
+ \phi_{i}(S + \epsilon_1 \tS_1+\epsilon_2 \tS_2)\frac{\partial^2}{\partial\epsilon_1 \partial \epsilon_2}\phi_{i}(S+\epsilon_1 \tS_1 \epsilon_2 \tS_2)=0,
\end{eqnarray*}
and setting $\epsilon_1=\epsilon_2=0$ we see that $Z_{i}(\tS_1,\tS_2)$ solves \eqref{SYLV-2}. Using  the linearity of the trace function and  \eqref{LEM-031} we see that 
\begin{eqnarray*}
[DF^{0}_{\M}(S)](\tS)=\left . \frac{d}{d\epsilon}F_{\M}(S+\epsilon \tS) \right\vert_{\epsilon=0} &=& Tr\left(\tS  -2 \sum_{i=1}^{n} w_{i} [D\phi_i](S)(\tS) \right)\\ 
&=& Tr\left(\tS - 2 \sum_{i=1}^{n} w_{i} Y_{i}(\tS)\right),
\end{eqnarray*}
where $Y_{i}$ is the solution of the Sylvester equation \eqref{SYLV-1}.  
We now take the second derivative of the Fr\'echet function and obtain that
\begin{eqnarray*}
[D^{2}F^{0}_{\M}(S)](\tS_1,\tS_2)= \left . \frac{\partial^2}{\partial\epsilon_1\partial \epsilon_2}F_{\M}(S + \epsilon_1 \tS_1+\epsilon_2 \tS_2) \right\vert_{\epsilon_1=\epsilon_2 =0} &=&- 2 Tr \left( \sum_{i=1}^{n} w_{i}[ D^2\phi_i (S)](\tS_1,\tS_2) \right)\\ 
&=& -2 Tr\left( \sum_{i=1}^{n} w_{i} Z_{i}(\tS_1,\tS_2) \right)
\end{eqnarray*}

Since $S_{B}$ is the minimizer of $F^{0}_{\M}$ and by definition $V_{\M}=\min_{S \in {\mathbb P}(d)} F^{0}_{\M}(S)$ it holds that $F^{0}_{\M}(S_B)=0$ and $[DF^{0}_{\M}(S_{B})](\tS)=0$ for all $\tS \in {\mathbb P}(d)$, by using Fermat's principle. By the fact that $S_{B}$ is the minimizer we see that $[D^2F^{0}_{\M}(S_{B})](\tS,\tS) \ge 0$ for every $\tS \in {\mathbb P}(d)$.

To see that the minimizer $S_{B}$ coincides with the solution of  \eqref{BAR-EQUATION} we use the results of \cite{alvarez2016fixed} or \cite{bhatia2017bures}).
\end{proof}

It turns out that $Tr(\sum_{i=1}^{n}w_{i} Y_{i})$ where $Y_{i}$ is the solution of the Sylvester equation \eqref{SYLV-1} can be expressed directly in terms of $S_{i}$, $i=1,\ldots, n$.  The following Lemma is essentially a result of  \cite{bhatia2017bures}, the proof of which is included here briefy for the ease of the reader.

\begin{lemma}[\cite{bhatia2017bures}]\label{LEMMA-04}  If $Y_{i}=Y_{i}(\tS)$ is the solution of the Sylvester equation \eqref{SYLV-1}, then 
\begin{eqnarray*}
Tr\left(\tS-2 \sum_{i=1}^{n}w_{i} Y_{i}(\tS) \right)=Tr\left( \left[I-\sum_{i=1}^{n} w_{i} S^{-1/2}(S^{1/2}S_{i}S^{1/2})^{1/2} S^{-1/2}\right]\tS \right).
\end{eqnarray*}
\end{lemma}

\begin{proof}
By the theory of the Sylvester equation (using also the fact that the matrices $S_{i},S \in {\mathbb P}(d)$) the solution $Y_{i}=Y_{i}(\tS)$ can be expressed as
\begin{eqnarray*}
Y_{i}=\int_{0}^{\infty} \exp(-t(S_{i}^{1/2}S S_{i}^{1/2})^{1/2}) S_{i}^{1/2} \tS S_{i}^{1/2} \exp(-t(S_{i}^{1/2}S S_{i}^{1/2})^{1/2}) dt
\end{eqnarray*}
and following \cite{bhatia2017bures} we see that using the properties of the trace (the invariance of the trace for any number of permutations of symmetric matrices) we see that
\begin{eqnarray*}
Tr(2 \sum_{i=1}^{n} w_{i} Y_{i})=\sum_{i=1}^{n} w_{i} \int_{0}^{\infty}  Tr\left( S_{i}^{1/2} \exp(-2 t(S_{i}^{1/2}S S_{i}^{1/2})^{1/2}) S_{i}^{1/2} \tS \right) dt  \\
=\sum_{i=1}^{n} w_{i} Tr\left(S_{i}^{1/2} \left(\int_{0}^{\infty}   \exp(-2 t(S_{i}^{1/2}S S_{i}^{1/2})^{1/2}) ) dt \right) S_{i}^{1/2} \tS \right).
\end{eqnarray*}
It holds that $\int_{0}^{\infty} e^{-t B} dt = B^{-1}$,  for any $B \in {\mathbb P}(d)$ (since $\int_{0}^{T} e^{t B} dt = B^{-1}(e^{T B}- I)$,  for non singular matrices $B$), so that
 $$\int_{0}^{\infty}   \exp(-2 t(S_{i}^{1/2}S S_{i}^{1/2})^{1/2}) ) dt =\frac{1}{2}(S_{i}^{1/2} S S_{i}^{1/2})^{-1/2}=\frac{1}{2}(S_{i}^{-1/2}S^{-1} S_{i}^{-1/2})^{1/2}, $$
so that
\begin{eqnarray*}
Tr(2 \sum_{i=1}^{n} w_{i} Y_{i})=\sum_{i=1}^{n} w_{i} Tr( S_{i}^{1/2} (S_{i}^{-1/2}S^{-1} S_{i}^{-1/2})^{1/2} S_{i}^{1/2} \tS )
=Tr(\sum_{i=1}^{n} w_{i} (S_{i}\# S^{-1})\tS),
\end{eqnarray*}
where $A\#B :=A^{1/2}(A^{-1/2}B A^{-1/2})^{1/2} A^{1/2}$ is the geometric mean of the matrices $A$ and $B$, or equivalently the unique positive  solution as the Riccati equation $\tS A^{-1} \tS=B$. Since the geometric mean is symmetric i.e., $A\# B=B\# A$, we have that 
\begin{eqnarray*}
Tr(2 \sum_{i=1}^{n} w_{i} Y_{i})=Tr(\sum_{i=1}^{n} w_{i} (S^{-1}\# S_{i})\tS)=Tr(\sum_{i=1}^{n} w_{i} S^{-1/2}(S^{1/2}S_{i}S^{1/2})^{1/2} S^{-1/2} \tS),
\end{eqnarray*}
from which we conclude that
\begin{eqnarray*}
[DF_{\M}^0(S)](\tS)=Tr ( [I - \sum_{i=1}^{n} w_{i} (S_{i}\# S^{-1})]\tS)=Tr ( [I - \sum_{i=1}^{n} w_{i} (S^{-1}\# S_{i})]\tS)\\
=Tr\left( \left[I-\sum_{i=1}^{n} w_{i} S^{-1/2}(S^{1/2}S_{i}S^{1/2})^{1/2} S^{-1/2} \right]\tS \right).
\end{eqnarray*}
\end{proof}

Now we have the required results to prove Proposition \ref{Prop2.8}.
\begin{proof}
Assume that the position of the firm is provided by the risk mapping $-X=\Phi_0(Z)$. Under the measure $\mu = LS(m,S)$, we have that ${\mathbb E}_{\mu}[-X]=\int_{\R^d} \Phi_0(m + S^{1/2} z) d\nu(z)$, where $\nu$ is the probability measure on $\R^d$ which characterizes the random variable $Z_0$ which generates the whole location - scatter family.  Using the definition of the function $\Phi  : \R^{d} \times \R^{d\times d} \to \R_{+}$ we conclude that $\RW(X)$  is given by the solution of the optimization problem \eqref{MATRIX-WAS-OPT}. By the strict convexity of the  function $F_{\M}=\overline{F}_{\M}+F^{0}_{\M}$, this problem is well posed for small enough values of $\gamma>0$, regardless of the form of the risk mapping $\Phi_0$, whereas if $\Phi$ is concave it is well posed for any $\gamma>0$.

We derive the first order conditions for problem \eqref{MATRIX-WAS-OPT}.  Since we have already computed the Fr\'echet derivative of the penalty function $F_{\M}$  in Lemma \ref{FR-FUN-DER}  we  just need to calculate the Fr\'echet derivative of the function $\Phi$. For any $(\tm,\tS) \in \R^{d}\times {\mathbb P}(d)$ by definition   $[D\Phi(m,S)](\tm,\tS) := \left . \frac{d}{d\epsilon} \Phi(m+\epsilon \tm, S + \epsilon \tS) \right\vert_{\epsilon =0} $. Assuming the necessary smoothness of $\Phi$ and using Lebesgue's dominated convergence theorem we may calculate
\begin{eqnarray*}
[D\Phi(m,S)](\tm,\tS) = \sum_{\ell=1}^{d} \Phi_{\ell}(m,S)  \tm_{\ell}  + \frac{1}{2} \sum_{\ell=1}^d \sum_{k=1}^d \left[ \left( \Psi S^{-1/2} + S^{-1/2} \Psi^{T} \right) \tS  \right]_{\ell, k},
\end{eqnarray*}
where we have defined the functions $\Phi_{\ell}, \Psi_{\ell k} : \R^{d} \times {\mathbb P}(d) \to \R$  as in \eqref{NEW-FUN-DEF}. Using the notation in \eqref{NOTATION-0} we see that we may express $[D\Phi(m,S)](\tm,\tS)$ in terms of
\begin{eqnarray*}
[D\Phi(m,S)](\tm,\tS)= \langle D_{m}\Phi(m,S), \tm\rangle + Tr ( D_{S}\Phi(m,S) \tS).
\end{eqnarray*}
The first order conditions for problem \eqref{MATRIX-WAS-OPT} become
\begin{eqnarray*}
\langle  D_{m}\Phi(m,S)  , \tm \rangle+ Tr ( D_{S}\Phi(m,S) \tS)+[D\overline{F}_{\M}(m)](\tm)+ [DF^{0}_{\M}(S)](\tS)=0, \,\,\, \forall \,(\tm,\tS) \in \R^{d}\times {\mathbb P}(d).
\end{eqnarray*}
Using Lemma \ref{FR-FUN-DER} and Lemma \ref{LEMMA-04}, we see that the first order conditions for problem \eqref{MATRIX-WAS-OPT} become
\begin{eqnarray*}
\langle D_{m}\Phi(m,S),\tm\rangle- \frac{1}{\gamma}\langle m-m_{B},\tm\rangle=0, \,\,\, \forall \, \tm \in \R^{d}, \\
Tr\left( \left( D_S\Phi(m,S) - \frac{1}{2\gamma}\left[ I - \sum_{i=1}^n w_i S^{-1/2}(S^{1/2} S_i S^{1/2})^{1/2} S^{-1/2} \right] \right) \tS \right) = 0, \,\,\, \forall \, \tS \in {\mathbb P}(d).
\end{eqnarray*}
However, the above conditions can be further reduced to the more convenient form
\begin{eqnarray*}
D_{m}\Phi(m,S)- \frac{1}{\gamma} (m-m_{B})=0, \\
D_{S}\Phi(m,S) - \frac{1}{2\gamma} \left( I -\sum_{i=1}^{n} w_{i} S^{-1/2}(S^{1/2}S_{i}S^{1/2})^{1/2} S^{-1/2} \right)=0,
\end{eqnarray*}
whereby multiplying the second from the left and from the right by $S^{-1/2}$ yields the stated result.
\end{proof} 

\subsection{Proof of Remark \ref{Rem2.9}}\label{A.4}
\begin{proof}
In fact, we may use a perturbative  expansion of the form $m=m_{B} + \gamma \tm + \ldots $, and $S=S_{B} + \gamma \tS + \ldots$  where $(\tm,\tS) \in \R^{d} \times {\mathbb P}(d)$ are corrections to be determined.
\\
Once such an expansion for the maximizer is available, then an expansion for the risk measure $\RW(X)$ can be obtained. Using Lemma \ref{FR-FUN-DER} we see that up to first order in the parameter $\gamma$, we have the expansion
\begin{eqnarray*}
\RW(X)=\Phi(m_{B},S_{B}) + \\
\gamma \left( ( D_{m}\Phi(m_ {B},S_{B}),\tm) +Tr( D_{S} \Phi(m_{B},S_{B}) \tS) - \frac{1}{4} [D^2\overline{F}_{\M}(m_{B})](\tm,\tm) - \frac{1}{4} [D^2F^{0}_{\M}(S_{B})](\tS,\tS) \right).
\end{eqnarray*}
Again using Lemma \ref{FR-FUN-DER} we have that
$ [D^2\overline{F}_{\M}(m_{B}) ](\tm,\tm)=2 \| \tm \|^2$ and  $ [D^2 F^0_{\M}(m_{B}) ](\tS,\tS)=-2 Tr (\sum_{i=1}^{n} w_{i}\Z_{i})$ ,
where  the matrices $\Z_i$, $i=1,\dots, n$ are solutions of the system of   Sylvester equations
\begin{eqnarray*}
\Y_{i} (S_{i}^{1/2} S_{B} S_{i}^{1/2})^{1/2} +  (S_{i}^{1/2} S_{B} S_{i}^{1/2})^{1/2} \Y_{i} = S_{i}^{1/2} \tS S_{i}^{1/2}, \,\,\, &&i=1,\ldots, n \\
\Z_{i} (S_{i}^{1/2} S_{B} S_{i}^{1/2})^{1/2} +  (S_{i}^{1/2} S_{B} S_{i}^{1/2})^{1/2}\ Z_{i} = - 2\Y_{i}^2, \,\,\, &&i=1, \ldots, n.
\end{eqnarray*}
Note that the system is linear as the first half of it is uncoupled from the second half. Substituting this into the above equation we have that
\begin{equation}\label{risk-measure-correction}
\begin{aligned}
\RW(X)=\Phi(m_{B},S_{B}) + \\
\gamma \left( ( D_{m}\Phi(m_ {B},S_{B}),\tm) +Tr( D_{S}\Phi(m_{B},S_{B})\tS) - \frac{1}{2} \| \tm \|^2 + \frac{1}{2} Tr\left( \sum_{i=1}^{n} w_{i}\Z_{i} \right)\right),
\end{aligned}
\end{equation}
\\
It thus remains to determine the corrections $(\tm,\tS) \in \R^{d} \times {\mathbb P}(d)$. 
Substituting this expansion in the first equation of \eqref{FOC-A}, assuming sufficient smoothness for $\Phi$, and separating powers  of $\gamma$ we see that 
\begin{eqnarray}\label{mean-correction}
\tm=D_{m}\Phi(m_{B},S_{B}). 
\end{eqnarray}
To obtain the correction for $S_{B}$, we substitute the expansion in the second equation of \eqref{FOC-A} using the Taylor expansion of the matrix square function, according to which $S^{1/2}=(S_{B}+\gamma \tS)^{1/2} \simeq S_{B}^{1/2} + \gamma \J(\tS)$, where $\J=\J(\tS)$ is the solution of the Sylvester equation $\J S_{B}^{1/2}+S_{B}^{1/2} \J=\tS$. We then see that to first order in $\gamma$, we have that $$(S^{1/2}S_{i}S^{1/2})^{1/2}=\left( S_{B}^{1/2} S_{i} S_{B}^{1/2} + \gamma ( \J S_{i} S_{B}^{1/2}+ S_{B}^{1/2}S_{i} \J) \right)^{1/2}$$
and Taylor expanding once more $(S^{1/2}S_{i}S^{1/2})^{1/2}=(S_{B}^{1/2} S_{i} S_{B}^{1/2})^{1/2} + \gamma \H_{i}$, where $\H_{i}$ solves the Sylvester equation  $\H_{i} (S_{B}^{1/2} S_{i} S_{B}^{1/2})^{1/2}+(S_{B}^{1/2} S_{i} S_{B}^{1/2})^{1/2} \H_{i}= (\J S_{i} S_{B}^{1/2}+ S_{B}^{1/2}S_{i} \J)^{1/2}$.  We subsitute these expansions into the second equation of \eqref{FOC-A} and see that the correction $\tS$ can be found by the solution $(\tS, \J,\H_{1},\ldots, \H_{n})$ of the (linear) system of matrix equations
\begin{eqnarray*}
\tS -\sum_{i=1}^{n}w_{i} \H_{i} = 2 S_{B}^{1/2} D_{S}\Phi(m_{B},S_{B}) S_{B}^{1/2}, \\
\tS - \J S_{B}^{1/2} - S_{B}^{1/2} \J =0, \\
\J S_{i} S_{B}^{1/2} + S_{i} S_{B}^{1/2} \J - \H_{i} (S_{B}^{1/2} S_{i} S_{B}^{1/2})^{1/2}-(S_{B}^{1/2} S_{i} S_{B}^{1/2})^{1/2} \H_{i} =0, \,\,\, i=1,\ldots, n.
\end{eqnarray*}
Note that this system has a sparse structure, which allows for its treatment in terms of iterative schemes. Upon substituting the above in \eqref{risk-measure-correction} we conclude that up to first order in $\gamma$, we have
\begin{eqnarray*}
\RW(X)= \Phi(m_{B},S_{B}) + \gamma \left(\frac{1}{2} \| D_{m}\Phi(m_{B},S_{B}) \|^2 +  Tr(D_{S}\Phi(m_{B},S_{B}) \tS) + \frac{1}{2} Tr\left( \sum_{i=1}^{n} w_{i} \Z_{i}\right) \right)
\end{eqnarray*}
\end{proof}
\subsection{Proof of Proposition \ref{Prop2.15} }\label{A.5}
\begin{proof}
Let us denote by $f$ the density related to the probability measure $\mu$ and $f_i$ denote the densities related with the priors in $\M$. The calculation of $\RE(X)$ is reduced to the solution of the minimization problem $\min_{f \in \mathcal{F} } \int_{\R^d} V( z, f(z) ) \ud z$, where $V(z, f(z) = \Phi_0(z) f(z) + \frac{1}{\gamma}\sum_{i=1}^n w_i f(z)\log\frac{f(z)}{f_i(z)}$ for an appropriate space $\mathcal{F}$ for the probability densities $f$, which is well-posed for any bounded set $\mathcal{G}\subset L^1(\R^d)$ which is also uniformly integrable. Notice that uniform integrability is guaranteed by De La Vall\'ee Poussin criterion, since there exists a function $H:\R_+ \to \R_+$ satisfying $\frac{H(u)}{\|u\|}\to +\infty$ as $\|u\|\to +\infty$ and such that $\sup_{f\in\mathcal{G}} \int H(f(u)) \ud u < +\infty$. Therefore, by Dunford-Pettis theorem, the set $\mathcal{G}\subset L^1(\R^d)$ is relatively compact for the weak topology of $L^1$ guaranteeing the well posedness of the problem. We also need to take into account the constraint that $f$ is a  probability density i.e., $\int_{\R^d}f(z) dz =1$. To take this into account we introduce a Lagrange multiplier $\lambda$ and consider the minimization problem $\min_{f \in L^{1+\epsilon}(\R^d)} \int_{\R^d}( V(z, f(z)) + \lambda f(z) ) dz$. We solve the problem for the optimal $f$. Let us denote $f_{G}:=\prod_{i=1}^{n} f^{w_i}_{i}$. A standard argument allows us to obtain the first order condition as $\gamma \Phi_0(z) + \gamma\lambda +1 + \log f(z) - \sum_{i=1}^n w_i \log f_i(z) = 0, \,\,\, z \,\, -a.e.$,
which leads to $f(z)=e^{\gamma \Phi_0(z) - \gamma\lambda - 1}f_{G}(z)$. From the normalization constraint $\int_{\R^d} f(z)\ud z = 1$ we can eliminate the associated Lagrange multiplier and obtain the density form $f_{\gamma}(z)=C_{\gamma}e^{\gamma \Phi_0(z)}f_{G}(z)$, where $C_{\gamma} := \left( \int_{\R^d} e^{\gamma \Phi_0(z)} f_{G}(z) \ud z \right)^{-1}$.
Knowing explicitly the optimal density function leads to the explicit calculation of the weighted entropic risk measure.
Note that at the critical point
\begin{eqnarray*}
\begin{aligned}
V(z, f_{\gamma}(z))
=-\Phi_0(z) e^{\gamma \Phi_0(z)} C_{\gamma} f_{G}(z) + \frac{1}{\gamma} \sum_{i=1}^n w_i \left\{ e^{\gamma \Phi_0(z)} C_{\gamma} f_{G}(z) \log\left( \frac{e^{\gamma \Phi_0(z)} C_{\gamma} f_{G}(z)}{f_i(z)}\right) \right\}\\
= -\Phi_0(z) e^{\gamma \Phi_0(z)} C_{\gamma} f_{G}(z)+\frac{1}{\gamma} \left\{ e^{\gamma \Phi_0(z)} C_{\gamma} f_{G}(z)\left( \log C_{\gamma} +\gamma \Phi_0(z) + \log f_G(z) - \sum_{i=1}^n w_i \log f_i(z)\right) \right\} \\
= -\Phi_0(z) e^{\gamma \Phi_0(z)} C_{\gamma} f_{G}(z) +\Phi_0(z) e^{\gamma \Phi_0(z)} C_{\gamma} f_{G}(z) -\frac{1}{\gamma} e^{\gamma \Phi_0(z)} C_{\gamma} f_{G}(z) \log C_{\gamma}= -\frac{1}{\gamma} e^{\gamma \Phi_0(z)} C_{\gamma} f_{G}(z) \log C_{\gamma}
\end{aligned}
\end{eqnarray*}
keeping in mind that $\int_{\R^d} e^{\gamma \Phi_0(z)} C_{\gamma} f_{G}(z) \ud z=1$. It is easy to check that the minimizer of $V_{\M}$ is the probability density function of the KL-barycenter, i.e. $f_0 = C_0 f_G$ where $C_0 = ( \int f_G dm )^{-1}$. Substituting this density in $V_{\M}$ we obtain that $V_{\M} = \log C_0$. Combining the above, we derive that
\begin{eqnarray*}
\RE(X) = \int_{\R^d} V(z ,f_{\gamma}(z)) \ud z + \frac{1}{\gamma} V_{\M}
= -\frac{1}{\gamma} \int_{\R^d} e^{\gamma \Phi_0(z)} C_{\gamma} f_{G}(z) \log C_{\gamma} \ud z + \frac{1}{\gamma} \log C_0\\
= \frac{1}{\gamma} ( \log C_0 - \log C_{\gamma} ) = \frac{1}{\gamma} \log \left( \frac{C_0}{C_{\gamma}}\right) = \frac{1}{\gamma} \log \left( \int_{\R^d} e^{\gamma \Phi_0(z)} f_0(z) dz\right) = \frac{1}{\gamma}\log( \mathbb{E}_{\mu_0}[e^{-\gamma X}] ).
\end{eqnarray*}
This concludes the proof.
\end{proof}
\subsection{Proof of Proposition \ref{Prop3.1} }\label{A.7}
\begin{proof}
We will use expansion \eqref{Rem2.9} twice, once for the risk position $X(\epsilon)$ and once for the risk position $X$. We will also use the following simplified notation \eqref{NOTATION-123} and \eqref{NOTATION-SIM}. The correction to the mean (around the mean barycenter $m_{B}$) for risk positions $X(\epsilon)$ and $X$ will be denoted by $\tm(\epsilon)$ and $\tm$ respectively, whereas the corection to the covariance (around the  covariance barycenter $S_{B}$) for risk positions $X(\epsilon)$ and $X$ will be denoted by $\tS(\epsilon)$ and $\tS$ respectively. Concerning the correction to the mean we have that $\tm(\epsilon)=M + \epsilon M_{j}$, so that $\tm'(\epsilon):=\frac{d}{d\epsilon} \tm(\epsilon)=M_{j}$. Concerning the correction to the covariance, the analysis is slightly more complicated but fortunately reduces to the solution of a sparse linear system of matrix equations.
Let $\tS(\epsilon)$ be part of the solution  $\tS(\epsilon), J(\epsilon), H_{1}(\epsilon), \ldots, H_{n}(\epsilon))$, of the system of matrix equations
\begin{equation}\label{RA-1}
\begin{aligned}
\tS(\epsilon)-\sum_{i=1}^{n}w_{i} \H_{i}(\epsilon) = 2 B(C + \epsilon C_{j}) B, \\
\tS(\epsilon)-\J(\epsilon) B - B \J(\epsilon)=0, \\
\J(\epsilon) D_{i} + D_{i} \J(\epsilon) -\H_i(\epsilon) E_{i} -E_{i} \H_{i}(\epsilon) =0,
\end{aligned}
\end{equation}
and $(\Y_1(\epsilon),\ldots, \Y_n(\epsilon))$ be the solution of the (decoupled) Sylvester equations
\begin{eqnarray}\label{RA-2}
\Y_{i}(\epsilon)G_{i} +G_{i} \Y_{i}(\epsilon) = B_{i} \tS(\epsilon) B_{i}, \,\,\, i=1,\ldots, n,
\end{eqnarray}
and $(\Z_1(\epsilon),\ldots, \Z_n(\epsilon))$ be the solution of the (decoupled) Sylvester equations
\begin{eqnarray}\label{RA-3}
\Z_{i}(\epsilon) G_{i} + G_{i} \Z_{i}(\epsilon) = -2 \Y_{i}(\epsilon)^2, \,\,\, i=1,\ldots, n.
\end{eqnarray}
Note that we may consider \eqref{RA-1}, \eqref{RA-2} and \eqref{RA-3} as one large sparse system of matrix equations.

We differentiate the above system with respect to $\epsilon$. To ease notation we will denote the derivatives with respect to $\epsilon$ by a prime, i.e., $\tS'(\epsilon)=\frac{d}{d\epsilon}\tS(\epsilon)$ with a similar notation for the other matrices as well.
We then obtain that
\begin{equation}\label{RDP-1}
\begin{aligned}
\tS'(\epsilon)-\sum_{i=1}^{n}w_{i} \H'_{i}(\epsilon) = 2 B  C_{j}  B, \\
\tS'(\epsilon)-\J'(\epsilon) B - B \J'(\epsilon)=0, \\
\J'(\epsilon) D_{i} + D_{i} \J'(\epsilon) - \H'_i(\epsilon) E_{i} -E_{i} \H'_{i}(\epsilon) =0,\\
\Y'_{i}(\epsilon)G_{i} +G_{i} \Y'_{i}(\epsilon) = B_{i} \tS'(\epsilon) B_{i}, \,\,\, i=1,\ldots, n,\\
\Z'_{i}(\epsilon) G_{i} + G_{i} \Z'_{i}(\epsilon) = -2 \Y'_{i}(\epsilon) \Y(\epsilon)-2 \Y'_{i}(\epsilon) \Y_{i}(\epsilon), \,\,\, i=1,\ldots, n.
\end{aligned}
\end{equation}

We now have that using \eqref{EXPANSION-000} with the notation \eqref{NOTATION-SIM} that
\begin{equation}\label{EXPANSION-000-EPS}
\begin{aligned}
\RW(X(\epsilon))= \Phi(m_{B},S_{B},\epsilon)  \\+ \gamma \left(\frac{1}{2} \| D_{m}\Phi(m_{B},S_{B},\epsilon) \|^2 +  Tr(\Phi_{S}(m_{B},S_{B},\epsilon) \tS(\epsilon)) + \frac{1}{2} Tr( \sum_{i=1}^{n} w_{i} \Z_{i}(\epsilon)) \right)= \\
A + \epsilon A_{j} + \gamma \left(  \frac{1}{2} \| M + \epsilon M_{j} \|^2 + Tr(( C + \epsilon C_{j}) \tS(\epsilon)) +\frac{1}{2} Tr( \sum_{i=1}^{n} w_{i} \Z_{i}(\epsilon)) \right),
\end{aligned}
\end{equation}
so that upon differentiation with respect to $\epsilon$ we have that
\begin{eqnarray*}
\frac{d}{d\epsilon} \RW(X(\epsilon))= A_{j} + \gamma ( M, M_{j}) + Tr (C_{j} \tS(\epsilon) + \epsilon C_{j} \tS'(\epsilon)) + \frac{1}{2} Tr (\sum_{i=1}^{n} w_{i} \Z'_{i}(\epsilon) ),
\end{eqnarray*}
and setting $\epsilon =0$ we have that
\begin{eqnarray*}
\RW(X_{k} \mid X) = A_{j} + \gamma (M,M_{j}) + Tr(C_{j} \tS) + \frac{1}{2} Tr(\sum_{i=1}^{n} w_{i} \Z'_{i}(0)) ),
\end{eqnarray*}
where $\Z_{i}':=\Z'_{i}(0)$ solves the system of matrix equations
\begin{equation}\label{RDP-2}
\begin{aligned}
\tS'-\sum_{i=1}^{n}w_{i} \H'_{i} = 2 B  C_{j}  B, \\
\tS'-J' B - B J'=0, \\
\J' D_{i} + D_{i} \J' - \H'_i E_{i} -E_{i} \H'_{i} =0,\\
\Y'_{i} G_{i} + G_{i} \Y'_{i} = B_{i} \tS' B_{i}, \,\,\, i=1,\ldots, n,\\
\Z'_{i} G_{i} + G_{i} \Z'_{i} = -2 \Y'_{i} \Y_i-2 \Y'_{i} \Y_{i}, \,\,\, i=1,\ldots, n.
\end{aligned}
\end{equation}
\end{proof}


\bibliographystyle{agsm}
\bibliography{frechet-risk-measures-multiple-sources}

\end{document}